\documentclass[reqno,10pt]{amsart}
\usepackage[normalem]{ulem}
\usepackage[all]{xy}
\usepackage[bookmarksnumbered,colorlinks]{hyperref}

\usepackage{amsfonts}
\usepackage{amsthm}
\usepackage[dvips]{graphicx,color}
\usepackage{graphicx}

\usepackage{xcolor}
\usepackage{amsmath}
\usepackage{indentfirst}
\usepackage[dvips]{graphicx,color}

\usepackage{hyperref}

\def\bbr#1\ebr{\textcolor{brown}
{#1}} %
\newcommand{\bb}{}
\newcommand{\eb}{}

\newcommand{\br}{}
\newcommand{\er}{}

\newcommand{\NN}{\mathbb{N}}
\newcommand{\C}{\mathbb{C}}

\newcommand{\R}{\mathbb{R}}

\def\be{\begin{eqnarray}}
\def\ee{\end{eqnarray}}
\def\beq{\begin{equation}}
\def\eeq{\end{equation}}

\def\R{{\mathbb R}}

\def\({\left (}
\def\){\right )}

\newtheorem{theorem}{Theorem}[section]
\newtheorem{lemma}[theorem]{Lemma}
\newtheorem{proposition}[theorem]{Proposition}
\newtheorem{corollary}[theorem]{Corollary}

\newtheorem{remark}[theorem]{Remark}
\newtheorem{example}[theorem]{Example}

\def\thk{\hat\theta_k}
\def\th0{\hat\theta_0}

\title[The Ehlers-Kundt conjecture]{The Ehlers-Kundt conjecture about\\  Gravitational Waves and Dynamical Systems 
}

\author[J.L. Flores]{Jos\'e L. Flores}
\address{Departamento de \'Algebra, Geometr\'{i}a y Topolog\'{i}a \hfill\break\indent Facultad de Ciencias, Universidad de M\'alaga \hfill\break\indent Campus Teatinos, 29071 M\'alaga, Spain.}
\email{floresj@uma.es}

\author[M. S\'anchez]{Miguel S\'anchez}
\address{Departamento de Geometr\'{\i}a y Topolog\'{\i}a \hfill\break\indent Facultad de Ciencias,
 Universidad de Granada\hfill\break\indent
 Campus Fuentenueva s/n,
 \hfill\break\indent 18071 Granada, Spain}
\email{sanchezm@ugr.es}

\date{12/06/2017.  
}

\begin{document}

\maketitle

	\begin{abstract}
\noindent
The Ehlers-Kundt conjecture is a physical assertion about  the fundamental role of plane waves for the description of gravitational waves.
Mathematically, it becomes equivalent to  a problem on the Euclidean plane $\R^2$ with a very simple formulation in Classical Mechanics: given a non-necessarily autonomous potential $V(z,u)$, $(z,u)\in\R^2\times\R$, harmonic in $z$ (i.e. source-free), the  trajectories of its associated dynamical system $\ddot{z}(s)=-\nabla_z V(z(s),s)$ are complete (they live  eternally) if and only if $V(z,u)$ is a  polynomial in $z$ of degree at most $2$ (so that  $V$ is a standard mathematical idealization  of vacuum).
Here, the conjecture is  solved in the \br significant case \er that $V$ is bounded polynomially in $z$ for $u$ in bounded intervals. The
 mathematical and physical
 implications of this {\em polynomial EK conjecture}, as well as the non-polynomial one, are discussed beyond their original scope. 

	\end{abstract}

\tableofcontents


\newpage

\section{Introduction}

\subsection{Physical viewpoint}

Setting the mathematical foundations of gravitational waves (certainly less known by the general public than its experimental search), required a tremendous intellectual battle. Among others,  substantial contributors include  Einstein \cite{Ei}, who predicted them in the framework of General Relativity (after its  more speculative introduction by Poincar\'e) by  developing a celebrated quadruple formula, Robertson, who clarified flaws in the use of coordinates by Einstein and Rosen, Bondi, who announced in \cite{B}
the discovery of a singularity-free solution of a plane
gravitational wave that carries energy,
Pirani \cite{Pi}, who linked them to curvature by appealing to mathematical tools by Synge, Petrov and Lichnerowicz, or Trautman \cite{Tr}, who defined the
boundary conditions to be imposed on gravitational waves at infinity.  The development of the theory was stimulated in a historic meeting at White Chapel'57, where Feynman explained his known physical {\em sticky bead} argument, and the existence of gravitational waves was accepted by the mainstream of relativists   at the end of the fifties\footnote{It is worth pointing out that the Bondi-Pirani-Robinson solution \cite{BPR},  obtained with great effort by physicists, had been
 studied  in the twenties by the
mathematician Brinkmann \cite{Br}, who introduced a class of metrics  including
pp-waves. We recommend the very recent exposition \cite{Nurowski} of the
historic development and foundations of the  mathematical theory;
reference \cite{CGS} is also recommended for both, the theoretical and experimental
history.}  \cite{WW, BPR}. However, a loose end has remained open since that heroic epoch.
In their well-known paper \cite{EK} published in 1962, Ehlers and Kundt proved
that plane waves are (geodesically) complete and posed the following problem ({\em the EK conjecture}, in what follows):

\begin{quote}
{\em Prove the plane waves to be the only
complete {\em [gravitational]} pp-waves, no matter
which topology one chooses}.
\end{quote}
According to these authors, complete and Ricci flat pp-waves would
represent a  graviton field independent of any matter by which it would be generated. Indeed,  they would admit neither a material source (by Einstein equation with zero cosmological constant) nor an external source (as completeness forbids extendibility \cite[p. 155]{O}). So, they would correspond to source-free photons in electrodynamics.
Recall that these photons are represented by  monochromatic sine waves and
constitute a useful  idealization  (the basis of  the Fourier analysis of homogeneous electromagnetic waves). The EK conjecture assigns a similar role to gravitational plane waves.

The EK conjecture is also related to the following {\em remarkable property} of plane waves pointed out shortly after by Penrose \cite{Pe}: {\em no spacelike
hypersurface exists
in the spacetime which is adequate for the global specification
of Cauchy data}; that is, they are not globally hyperbolic. Because of this property, Penrose wondered at what extent plane waves were physically meaningful or just  idealizations of the
model. The EK conjecture assigns such an idealized role to gravitational plane waves; indeed, it implies that more realistic vacuum pp-waves
(which may be also non-globally hyperbolic)\footnote{A detailed study of global hyperbolicity and other causal properties of pp-waves can be found in \cite{FS_JHEP}. It is worth pointing out that there are globally hyperbolic pp-waves which are ``close'' (in some appropriate mathematical sense) to plane waves; this opens another way out to Penrose's question, \cite{FS}.} must be incomplete, suggesting that a source might have been missed in the construction of the model.

\subsection{Mathematical formulation in potential theory}
Let us pose EK conjecture in a mathematically precise framework\footnote{This is especially important here because the conjecture was posed in a rather informal way, and this may have misleaded or discouraged to some researchers. On the one hand, the mentioning of topology is unclear and, on the other, the assumption {\em gravitational} was not included  explicitly in its statement (even though it could be deduced from the context). Indeed, without this hypothesis,  counterexamples would appear directly from the item (a)
below (see also
\cite[p. 25]{HR}, \cite[p. 70]{Bi}).
}. Any {\em pp-wave} (parallelly propagated plane-fronted wave) can be written as $\R^4$ endowed with the Lorentzian metric:
\be \label{e_ppwave} g = dx^2 + dy^2 + 2\ du\ dv
+ H(z,u)\ du^2 , \quad z:=(x,y), \qquad (x, y,u,v) \in \R^4, \ee  where $H :
\R^2\times \R \to \R$ is a smooth function. A pp-wave is a {\em plane wave} when, at each $u\in \R$, the $u$-constant function $H(\cdot, u)$
\begin{equation}
\label{e_uconstant}
\R^2\ni z=(x,y)\mapsto H(z,u)
\end{equation}
is a polynomial
in $x,y$ of degree $\leq 2$. Following \cite{CRS}, the function $H$ is  called
{\em polynomially upper bounded along finite
$u$-times} or just
{\em polynomially $u$-bounded} when, for each $u_0\in \R$,
there exists $\epsilon_0>0$ and a polynomial $q_0$ on $
\R^2$  such that $H(z,u)\leq q_0(z)$ for all $(z,u)\in
\R^2\times (u_0-\epsilon_0,u_0+\epsilon_0)$; in this case,
$H$ is {\em
quadratically polynomially $u$-bounded} when $q_0$ can be
chosen of degree $2$ for all $u_0\in\R$.
A straightforward computation shows that $g$ is Ricci flat (i.e., the pp-wave is gravitational) iff $H$ is harmonic respect to its first pair of variables,  that is,
 \be \label{e_laplacian} \Delta_z H(z,u):= (\partial^2_xH+ \partial^2_yH)(z,u)\equiv 0. \ee
Putting $V=-H$   one obtains that a pp-wave  is geodesically complete if and only if all the inextendible solutions $I\ni s\mapsto z(s)\in \R^2$ for the dynamical system
\begin{equation}
\label{e_v}
\ddot z(s) = -\nabla_z V(z(s),s)
\end{equation}
are complete, i.e., defined on $I=\R$  (see Section \ref{s2.2}).
That is, the EK conjecture becomes fully equivalent, at least up to the issue of topology, to the following (recall that $V$ is lower bounded if $H$ is upper bounded):
\begin{quote}
{\em For any 
  potential $V(z,u), (z,u)\in\R^2\times \R$ such that the $u$-constant function 
$V(\cdot,u)$  is harmonic
 for each $u\in \R$, the trajectories of the dynamical system \eqref{e_v} are complete if and only if $V(\cdot,u)$ is an at most quadratic polynomial for each $u\in \R$.}
\end{quote}

\subsection{Known results}
Even though the completeness of trajectories of dynamical systems is  a very classical topic
\br (see
the book \cite{AM} or the survey \cite{Storonto}), \er as far as we know, the condition of harmonicity (which involves potential theory as well as  divergence free gradient vector fields) has not been studied in this setting.
Indeed, the progress along these decades has focused on other aspects of the conjecture rather than in the crude geodesic equation (or  the dynamical system \eqref{e_v}), concretely:
\begin{enumerate}
\item[(a)] \label{11} All plane waves (gravitational or not) are geodesically complete, as the equation \eqref{e_v} reduces to a second order  linear system of differential equations \cite{EK}, \cite[Ch. 13.1]{BEE}. Moreover, any  pp-wave such that $H=-V$ is quadratically polynomially $u$-bounded is complete too \cite[Theorem 2]{CRS}.

\item[(b)] \label{2} If a harmonic function $f$ on $\R^2$ is (upper)  bounded by a polynomial of degree $n$, then $f$ is itself a harmonic polynomial of degree at most $n$ (see Section~\ref{s2.1}). Therefore, the EK conjecture becomes true if $H$ is quadratically polynomially bounded, as such a bound implies directly that the pp-wave is a plane wave. There are some physical reasons supporting such a bound \cite{FS, FShonnef}; remarkably among them, the fact that otherwise the pp-wave would not be strongly causal at least in the autonomous case \cite[Corollary 1.3]{CFH}.

\item[(c)] As quoted above, Ehlers and Kundt's original statement of the conjecture   included the extra: ``no matter which topology one chooses''. This is somewhat
vague, as pp-waves are assumed to be defined on $\R^4$, and it is not obvious how to extend such a definition on a general manifold $M$
(say,  a non-trivial extension for the conjecture, so that $M$ were not covered automatically by a plane wave on $\R^4$). However, Leistner and Schliebner \cite{LS} gave a local characterization of pp-waves, argued  that a natural
extension of the conjecture follows when   a locally pp-wave metric is taken on a compact $M$, and proved that this extension becomes equivalent to the  standard one on $\R^4$. In this same line, the quoted article \cite{CFH} also showed that  plane waves are the universal coverings of certain spacetimes with non-trivial topology under some natural hypotheses for the EK conjecture.

\item[(d)] Recently, the completeness of {\em impulsive waves} has been  studied intensively \cite{SS1,SS2,SSS,SS4,SS3}. These waves have a non-continuous profile  type $H(z,u)=f(z)\delta(u)$ for some  (generalized)
delta-function $\delta$ and smooth $f$. Thus,   the function $H$ can be regarded as $z$-harmonic when $\Delta f=0$. The mentioned results of completeness yield  counterexamples to  the EK conjecture in the impulsive setting, showing the necessity of continuity in $u$ as well as the appropriate smoothness of $H$  (see Remark \ref{r_C1} below).

\end{enumerate}

\subsection{Main results and plan of work}
We will prove the following result about the solutions to \eqref{e_v}
and, thus, the {\em polynomial  EK conjecture}:

\begin{theorem}\label{t1} Let $V: \R^2\times \R \rightarrow \R$ be a  polynomially $u$-bounded
$C^1$-potential
which is also $C^2$ and harmonic  in its first pair of
variables $z=(x,y)$.  Then:
all the solutions to the  dynamical system \eqref{e_v} are complete if and only if the function  $V(\cdot, u)$ is an at most quadratic polynomial  for each $u\in \R$.
\end{theorem}

\begin{corollary}
The Ehlers-Kundt conjecture is true when the characteristic coefficient $H(z,u)$ of the pp-wave is polynomially $u$-bounded.
\end{corollary}

\begin{remark}\label{r_C1} {\em Recall that the natural order of smoothness in the non-autonomous case is $C^1$ for $u$ (in order to construct the spacetime Levi-Civita connection and geodesics)  and $C^2$ in $z$ (in order to impose harmonicity). The latter implies analyticity in $z$, however, the $z$-analyticity of $V$ will not be used explicitly even in the autonomous case.}
\end{remark}
However, our result is significant even in the autonomous case $V(z,u)\equiv V(z)$. Indeed,  although the optimal polynomial bound on $-V(z)$  to ensure completeness is known to be quadratic (recall item (a) 
above), harmonicity is strongly required to ensure incompleteness when such a bound is violated.

Technically, the polynomial  bound  in our results deserves an special attention. The harmonicity of each $H(\cdot , u), u\in\R$, implies its analyticity and, then, the possibility
to expand it as an infinite polynomial series ---which become finite  whenever $H(\cdot , u)$ is bounded by a polynomial, see the item (b) above. If there were a  ($u$-polynomially unbounded)  counterexample to the
EK conjecture, this would be surprising and might admit some interpretation as a sort of  exotic state of vacuum. However, the solved polynomial case  might also suggest that such a counterexample should be regarded just as a mathematical possibility rather than  a physical one (recall that many physical theories  are applied just until some perturbative order).
Nevertheless, whether  this open possibility is relativistically significant or not, its mathematical interest in the fields of dynamical systems and   Classical Mechanics remains  apparent.

The following digression emphasizes the links with complex analysis, concretely,  the relation between the autonomous case and the theory of completeness of {\em holomorphic  vector fields} on $\C^2$. Regarding $\R^2$ as the field of complex numbers $\C$, the appearance of an harmonic potential  suggests the complexification of the problem. The Lagrangian or Hamiltonian approaches  show the equivalence between the completeness of the trajectories for a potential
 $V$ on a manifold $M$ and the completeness of a vector field $X$ on the tangent or cotangent bundle \br (see \cite{AM, Storonto}). \er Such a bundle  can be identified with $\C^2$ in our case.
There exists a well-established theory about completeness of holomorphic vector fields $X$ on $\C^2$ in the case that they are {\em polynomial} (see the review \cite{LM} or more recent works such as \cite{BG, Bru}); such a theory includes links to the real analytic case (see for example \cite{Forst}). The usage of complex variables allows us to obtain very powerful results in comparison with those obtained just by means of the underlying real variables\footnote{A remarkable case is the following one obtained by Rellich \cite{Re}: If $h=h(z)$ is an entire function that is not an affine function
(i.e. a polynomial of degree at most 1), then every entire solution
of the complex analytic differential equation z''=h(z) is constant.}. However,  not all
the real analytic problems admits the usage of complex techniques, and unfortunately, the EK conjecture seems to be one of them. In any case,  both, the results for harmonic potentials studied here and those for holomorphic vector fields, require the hypothesis of  polynomial behavior, and  the non-polynomial case  remains as an open challenge for both, the analytic  and the complex cases. So a solution of the remaining open case of the EK conjecture might have relevant applications for the complex theory too.

The remainder of the paper is organized as follows. In Section \ref{s2}, the framework of harmonic polynomials and dynamical systems is provided. Then, an intuitive idea of the proof (which may serve as a guide for the reader) is sketched. The material in this part is either summarized from the literature or solved by elementary methods. However,  it contains some technical questions which must be taken into account along the proof (for example, Remark \ref{r_polinomio_no_autonomo}). The proof of both, the  autonomous and non-autonomous cases, is carried out in Section \ref{s3}.
More precisely, we will prove the autonomous case by means of Theorem \ref{zt}, and the non-autonomous generalization by Theorem \ref{t}. These two theorems identify some ``big'' sets of incomplete trajectories.  Indeed, such trajectories are shown to be confined in some regions (labeled $D_k[\rho_0, \theta_+]$) where estimates for incompleteness can be carried out.
Such a property of confinement (proved by detecting some principal radial directions for the higher order monomial of $V$ and estimating the amplitude of the oscillations of the trajectories around them) is the crux of the proof, and might have independent interest for orbits of dynamical systems.
In the autonomous case,  the proof is divided into three parts, each one in a subsection. In the non-autonomous case,   the same steps are followed, but
some technicalities might obscure the underlying ideas (see for example Remark \ref{dificultadnoautonomo}). So, the proof is revisited stressing the differences with the autonomous case.
We end with some conclusions and prospects in the last section.

\section{Setup}\label{s2}  

Even though the EK conjecture was posed in spacetimes, we will focus in its translation to dynamical systems and, so,  conventions in this field will be used. In order to avoid conflicts with  Relativity, we will use only $V(z,u)$ (instead of $-H(z,u))$, because  $V$ provides the intuition of potential energy and maintains the link with Classical Mechanics. The classical interpretation of the coordinate $u$ in $V(z,u)$ as an external {\em time} can also be maintained; however, we would like to emphasize that it is not by any means valid relativistically\footnote{When considered on the pp-wave, $u$ corresponds with a function on all $\R^4$ which is {\em not} a time function. Indeed, it does not grow on all future-directed causal curves  and it is just a {\em quasi-time} function, according to \cite{BEE} (see its
Definition 13.4, p. 490, and comments in p. 577). Because of this, the letter $t$ is avoided here.}. We will use always the name autonomous/non-autonomous to mean the $u$-independence/dependence of $V$, and the word {\em homogeneous} will be reserved for (harmonic) polynomials of degree $n$ (meaning that they do not contain terms of degree lower  than $n$)\footnote{Frequently, a pp-wave is called {\em homogeneus} (or, more properly, {\em locally symmetric}) in the unrelated sense that $H$ is independent of $u$.}.

 We refer to the original article \cite{EK} or the survey \cite[Section 8]{Bi} for additional background on relativistic waves.

\subsection{Harmonic polynomials}\label{s2.1}
Let us collect some well-known properties of any harmonic
function $V$   on $\R^2$ (say, the potential in the autonomous case), see for example \cite{ABR, Con}.

Regarding $z=(x,y)\in\R^2$ as $x+i y\in\C $,
harmonicity is equivalent to saying that $V$ is the real
part of an entire function $f(z)=V(z)+iW(z)$ on $\C\equiv \R^2$ \cite[p. 43]{Con}.  If the growth of $V$ is polynomially bounded with
degree $n$ (namely, $V(x,y) \leq A (x^2+y^2)^{n/2}$ for some $n\in \NN$, $A>0$ and large $x^2+y^2$) \br then $|V|$ is also polynomially bounded
and $V$ must be a
polynomial of degree at most $n$; indeed,  such a result is valid even for harmonic functions on $\R^N$, $N\in \NN$ (even though there are some natural links with Liouville's theorem on $\C^2$),
see \cite[Theorem II]{boas}. \er

In the case of harmonic polynomials on $\R^2$,
the homogeneous ones of degree $m>0$  constitute a 2-dimensional vector space \cite[Proposition 5.8]{ABR}.
Using polar coordinates $(\rho,\theta)$, such polynomials $p_m$ can be expressed as
\begin{equation}\label{e_base_homogeneous_harmonic}
p_m(\rho,\theta)= A_m\rho^m\cos(m\theta) + B_m\rho^m \sin(m\theta) =
\lambda_m \rho^m \cos m(\theta + \alpha_m)
\end{equation}
for unique $A_m,B_m\in \R$ and $\lambda_m> 0, \alpha_m\in (-\pi,
\pi]$. In Cartesian ones:
\begin{equation}\label{e1.5}
p_m(x,y)=\lambda_m h_m(R_m(x,y)),\;\; \hbox{where} \;\;
h_m(x,y):=\sum_{k=0}^{[m/2]}  (-1)^{k}  \binom{m}{m-2k}x^{m-2k}y^{2k},
\end{equation}
 for some rotation $R_m\in SO(2)$. So, any harmonic polynomial $p$ of degree $n>0$ can be written as $p=\sum_{m=0}^n p_m$  for some $p_0\in
\R$  and $p_m$, $m\in\{1,\dots , n\}$, as above, being all of them univocally determined.
 In particular, the potential $V$ in Theorem~\ref{t1} admits such an expression when it is autonomous. The non-autonomous case is somewhat subtler, but the following summary will be enough for our purposes.

\begin{proposition}\label{p_polinomio_no_autonomo} Let $V$ be a polynomially $u$-bounded potential. For each bounded interval $I\subset \R$, there exists a minimum nonnegative integer $n_0$ such that   $V$ is written as
\begin{equation}
\label{e_vpolar}
V(\rho,\theta,u)=-\lambda(u)\rho^{n_0}\cos n_0(\theta + \alpha(u))-\sum_{m=0}^{n_0-1}\lambda_m(u)\rho^{m}\cos m(\theta+\alpha_m(u)),
\end{equation}
(the first term of the right-hand side is distinguished for notational convenience) for all  $u\in I$, where the real functions $\br \lambda, \lambda_m \geq 0 \er $ ($m=0,\ldots,n_0-1$) are continuous and univocally
determined.

 Moreover, there exists a dense subset $\mathcal{D}\subset \R$,  each one of its points $u_0\in\mathcal{D}$ admitting some $
\epsilon_0> 0$ such
that the expression
\eqref{e_vpolar} obtained for $I=(u_0-\epsilon_0, u_0+\epsilon_0)$ satisfies additionally: $\lambda(u_0)\neq 0$ and all the  functions $\lambda, \lambda_m, \alpha, \alpha_m$ are $C^1$ on $I$.


In particular, in the autonomous case that expression is valid for $I= \R$ and all the functions $\lambda, \lambda_m$, $\alpha$, $\alpha_m$, can be regarded as real scalars with $\lambda\neq 0$.
\end{proposition}

\begin{proof}
As discussed above, the polynomic $u$-bound implies that $V(\cdot,u)$ is a  harmonic polynomial $p_{[u]}$ of degree $n(u)\in \NN$ for some (univocally determined) constants $A_m(u), B_m(u), m\leq n(u)$ in \eqref{e_base_homogeneous_harmonic}. Notice, however, that  $n(u)$ depends only lower semi-continuously on $u$.
Putting
$A_m(u)=B_m(u)=0$ for all $m> n(u)$, $A_m$ and $B_m$ can be regarded as functions on $\R$ for all $m$. As $n(u)$ must be locally
bounded  and the closure of $I$ is compact, a bound for $n(u)$ on all $I$ (and, then, the claimed minimum  bound $n_0$), can be found. Now, taking the derivatives $\partial^k_\rho V(\rho,\theta,u)$, $k\leq n_0$ and evaluating them at $\rho=1$, $\theta=0,\pi/2$, all the functions $A_m, B_m$
are shown to be $C^1$ on $I$ (and, thus, 
on all $\R$). From these functions, one constructs univocally the continuous functions $\lambda, \lambda_m \geq 0$, which are $C^1$ whenever they do not vanish, \br and chooses the functions $\alpha, \alpha_m$. \er

For the construction of $\mathcal{D}$ note that for each $\bar u\in \R$ there exists $
\epsilon>0$ such that the function $\lambda$ in the expression
\eqref{e_vpolar} obtained for $I_\epsilon=(\bar u-\epsilon,\bar u+ \epsilon)$ does not
vanish arbitrarily close to $\bar u$, in the sense that there exists $\{u_k\}\rightarrow
\bar u$ such that $\lambda(u_k)>0$ (indeed, recall that the value of $n_0$ for  $I_\epsilon$ must be constant when $\epsilon$ is small). So, by continuity $\lambda$ is positive (and, thus, $C^1$) in some open interval $I_k^{(1)}$ containing $u_k$ and $\alpha$ can be also chosen continuous (and, then, necessarily $C^1$) there. To assure  the regularity of the lower order $\lambda_m, \alpha_m$, put $m=n_0-1$. For each point $\bar u\in I_k^{(1)}$,  either there is a sequence $\{u_k^{(1)}\}\rightarrow
\bar u$ such that $\lambda_m(u_k^{(1)})>0$, (and we can reason for $\lambda_m , \alpha_m$, as in the previous case) or there is a neighborhood of  $\bar u$ where $\lambda_m$ vanish (and $\alpha_m$ can be chosen as a constant there). After $n_0$ steps, the required points are obtained.

 Finally the assertion on the autonomous case is straightforward.
 \end{proof}

 \begin{remark}\label{r_polinomio_no_autonomo}{\em
About the necessity of the set $\mathcal{D}$, notice that the possible non upper semi-continuous variation of the
 degree $n(u)$ of the polynomial $p_{[u]}=V(\cdot, u)$,  impedes to assure
 $ \lambda(u)\neq 0$ in general, even for small $\epsilon$.

 Moreover, $\alpha_m$ may not be continuous for any choice of $I$, even if $V$ is $C^\infty$. For example, define $V$ as $p_{[u=0]}=V(z,0)\equiv 0$ and, otherwise, put
 in \eqref{e_base_homogeneous_harmonic}:
 $$
\lambda_m(u)= e^{-1/u^2}, \qquad \alpha_m(u)= 1/u, \qquad \forall u\in \R\setminus\{0\}.
 $$
  }\end{remark}

\begin{remark}{\em
In the autonomous case, the at most quadratic harmonic polynomials are linear combinations of $1, x, y, x^2-y^2, xy$ and, so, the corresponding (conservative) forces are either constant forces (combinations of $\partial_x, \partial_y$) or linear combinations of  $x\partial_x-y\partial_y$ and $x\partial_y+y\partial_x$, which correspond to the forces of an harmonic oscillator (proportional to the distance and opposed to the motion) in one direction and the contrary forces  (proportional to the distance but in the direction to the motion)  in the orthogonal direction.
So, \br according to the EK conjecture, \er these should be  the unique conservative forces with no source whose dynamics become complete.

In the non-autonomous case, one has similar $u$-dependent forces which are quadratically polynomially $u$-bounded, even though the caution in the previous remark must be taken into account. Recall that the result which ensures its completeness (\cite[Theorem 2]{CRS}) does not require any expression of type \eqref{e_vpolar}.}
\end{remark}

\subsection{Correspondence with dynamical systems}\label{s2.2}

First, let us state how the completeness of a (non-necessarily gravitational) pp-wave becomes equivalent to  the completeness of a dynamical system.

\begin{lemma}\label{l_CFSwaves} Let $(\R^4,g)$ be a pp-wave as in \eqref{e_ppwave} with characteristic coefficient $H$ and consider the dynamical system \eqref{e_v}
associated to $V=-H$.

(1) The pp-wave is geodesically complete if and only if
the trajectories (inextendible solutions) for the dynamical system \eqref{e_v} are complete.

(2) The pp-wave is gravitational (i.e. Ricci flat) if and only if the potential $V(z,u)$ is harmonic in $z=(x,y)$.
\end{lemma}
\begin{proof} This is a straightforward computation from the geodesic equations and curvature of the pp-wave. Detailed computations in a more general framework are developed in \cite{CFS} (see Th. 3.2 and Prop. 2.1 therein to prove (1) and (2), resp.).
\end{proof}

\begin{remark}\label{r_v_simpliifcation}
{\em This lemma  reduces our problem to the classical framework of \br dynamical systems  \cite{AM, Jost}. \er
When necessary, some transformations can be carried out  to simplify the potential:

(a) Choose some constant $\lambda>0$ and change $V$ by the ``homothetic'' one $\lambda V.$

(b) Choose some rotation $R$ of $\R^2$ and change  $V(z,u)$ by $V(R(z),u)$.

(c) Replace $V$ by $V-p_0$, for some $p_0\in \R$.

\smallskip \noindent Obviously, neither the completeness nor the harmonic character of $V(\cdot , u)$ is changed by these transformations.
}\end{remark}
A standard conservation law establishes that, in the autonomous case,  the total energy (kinetic plus potential) must be preserved
along any trajectory $\gamma$ of \eqref{e_v}. In the non-autonomous case, the variation of the total energy is obtained by integrating $\partial_u V$ along the trajectory. If the trajectory is written in polar coordinates $\gamma(s)\equiv(\rho(s),\theta(s))$, $s\in I$, this is written:
\begin{equation}\label{e18}
\begin{array}{c}
\frac{1}{2}(\dot{\rho}(s)^2+\rho(s)^2\dot{\theta}(s)^2)+V(\rho(s),\theta(s),s)\\
=\frac{1}{2}\left(\dot{\rho}({s_0})^2 + \rho(s_0)^2\dot{\theta}(s_0)^2\right) + V(\rho({s_0}),\theta({s_0}),s_0)
 +\int_{{s_0}}^s\partial_u V(\rho(\sigma),\theta(\sigma),\sigma)d\sigma ,
\end{array}
\end{equation}
 for all $s_0, s\in I$.

The final aim of our technique will be to prove incompleteness by using the following simple result about solutions  $\rho$ to an ordinary differential inequation.

\begin{lemma} \label{l_ODEincomplete}
Let $\lambda: [0,\infty) \rightarrow \R$ be a $C^1$ function, choose any \br real number \er $n>2$, and consider the differential inequality:
\begin{equation}\label{e_ODEincomplete}
\ddot \rho(s) \geq  n\lambda \rho^{n-1}(s),
\qquad \qquad \rho(0)> 0, \qquad
\dot \rho(0)\geq 0,
\end{equation}
for all $s\geq 0$ in the maximal domain $I=[0,b), b\leq \infty$ of the $C^2$ function $\rho$.

(i) If $\lambda \geq \lambda_0$ for some constant $\lambda_0>0$, then all the solutions to \eqref{e_ODEincomplete} are incomplete (that is, $b<\infty$).  Moreover, for any $\kappa_0>0$ there exists $\rho_0>0$ such that any solution with initial condition $\rho(0)>\rho_0$ satisfies $b<\kappa_0$.

(ii) If $\lambda(0)>0$, then there exists some $k>0$ such that all the solutions to (\ref{e_ODEincomplete}) with either $\rho(0)>k$ or $\dot \rho(0)>k$  are incomplete.
\end{lemma}

\begin{proof}
In both cases, for some $\lambda_0, \epsilon>0$ one has $\lambda|_{[0,\epsilon]}\geq \lambda_0$ and, thus, $\dot\rho(s)>0$ on $(0,\epsilon]$. Composing  $\dot\rho(s)$ with $\rho^{-1}$ in that interval, one has $d(\dot\rho^2)/d\rho\geq 2 \lambda_0 n \rho^{n-1}$ and, integrating both sides twice with respect to $\rho$:
\begin{equation}\label{e_integral}
\int_{\rho(0)}^{\rho} \frac{d\bar \rho}{\sqrt{2\lambda_0(\bar\rho^n-\rho(0)^n) +\dot \rho(0)^2}} \geq s(\rho)
\end{equation}
whenever $s(\rho) \leq \epsilon$.
 In  the case (i), putting   $\rho=\infty$ the integral is finite  for any allowed $\rho(0), \dot\rho(0)$, because of the behavior of $\bar\rho^{n/2}$ at infinity for $n>2$ (this holds even when $\dot\rho(0)=0$ because, close to $\rho(0)$,  the integrand behaves
as $\eta/\sqrt{\bar\rho-\rho(0)}$
 for some $\eta>0$). So,  an upper bound for the possible values of the parameter $s$ is found.
Moreover, once $\kappa_0>0$ has been prescribed, the finiteness of the integral implies the existence of the required $\rho_0$.

In the case (ii), it is enough to choose initial conditions such that the integral in the left hand side of  \eqref{e_integral} is smaller than $\epsilon$ for $\rho=\infty$. Clearly, this happens by choosing
 either $\rho(0)$ or $\dot\rho(0)$ big enough.
\end{proof}

\begin{remark} {\em Lemma \ref{l_ODEincomplete} ensures incompleteness when  the right-hand side of \eqref{e_ODEincomplete} has a superlinear growth in $\rho$.
The hypothesis $\lambda\geq \lambda_0>0$ in the case (i) cannot be improved in $\lambda>0$ as, in particular, the next example shows.
}\end{remark}

\begin{example}{\em
Let $\rho_0: [0,\infty) \rightarrow \R_{>0}$ be any smooth positive function and put $n\lambda=\ddot \rho_0/\rho_0^{n-1}$ for some
$n>2$. Obviously, $\rho_0$ is  a complete solution of the ODE
$$
\ddot \rho = n \lambda \rho^{n-1} \qquad \qquad \rho(0)>0, \quad
$$
 and there are many choices of $\rho_0$  so that $\dot\rho_0(0)\geq 0$ and $\lambda >0$
(say, $\rho_0(s)\equiv e^s, 1+s+s^2...$).
However, the part (ii) of Lemma \ref{l_ODEincomplete} ensures that ``many'' incomplete solutions appear close to any $s_0\in \R$ such that $\lambda(s_0)>0$.
}\end{example}

\subsection{Heuristic argument and outline of proof} \label{s2.3}
In order to explain the idea of our technique, consider first
the simple case of an autonomous and homogeneous potential $V$
of degree $n>2$. Up to a rotation and an homothety,
we can assume
$$V(\rho,\theta)=-\rho^n\cos(n\theta).$$
Put $\theta_k=2\pi k/n$ for $k=0,1,\dots, n-1$. A simple computation shows that  the radial curve $\gamma_k(s)=(\rho(s), \theta(s)\equiv\thk)$ is a trajectory for $V$ whenever $\rho$ solves
the equality in \eqref{e_ODEincomplete} with $\lambda\equiv 1$
(check the trajectory equations \eqref{ze.14'} below). Therefore,
Lemma~\ref{l_ODEincomplete} (i) ensures that such a trajectory is necessarily incomplete.

What is more, one can check\footnote{This can be traced from our computations, see Remark \ref{r_simplificayrefinahomog}.} that there exists a radial
region, to be labeled $D_k[\rho_0, \pi/(2n)]$, around  each $\gamma_k$ such that
the trajectories starting  in $D_k[\rho_0, \pi/(2n)]$
(with suitable initial conditions) satisfy:
\begin{enumerate}
\item[(a)]  they remain in $D_k[\rho_0, \pi/(2n)]$, and

\item[(b)] whenever they remain in $D_k[\rho_0, \pi/(2n)]$,
the differential inequality in Lem\-ma~\ref{l_ODEincomplete} holds for its radial component (thus, that lemma  ensures incompleteness again).
\end{enumerate}
The existence of $D_k[\rho_0, \pi/(2n)]$ can be understood because $V$ remains decreasing and concave along each trajectory $\gamma_k$. Then, the harmonicity of the polynomial $V$ implies that, when only the $\theta$-coordinate is varied, each point of $\gamma_k$ becomes a strict minimum of $V$. This makes $\thk$ to behave as a stable equilibrium for the $\theta$-component of the trajectories close to $\gamma_k$.

In the case of a non-homogeneous potential, the curves $\gamma_k$ will not be trajectories for the potential because the terms of lower degree may force the $\theta$-component to be non constant. \br However, our aim will be to show that  some  regions, labeled $D[\rho_0, \theta_+]$, can be still found, so that they satisfy  the same qualitative properties (a) and (b) above (for suitable initial data in a smaller region $D[\rho_0, \theta_0]$). \er To this aim, the
 energies (kinetic, potential, radial)  of the $V$-trajectories will be compared with the corresponding ones for certain radial curves which will play the role of $\gamma_k$.
Accordingly, our task in the autonomous case will be carried out in three steps, each one in one of the first  subsections of Section \ref{s3}:
 \begin{enumerate} \item  For appropriate radial curves $\gamma_k$, the regions $D[\rho_0, \theta_+]$ are defined so that
certain technical bounds for $V$ and $\partial V/\partial \rho$ hold (Lemma \ref{zl1}).

\item For suitable trajectories $\gamma$ starting at $D[\rho_0, \theta_+]$,  the  growth of the amplitude  of each oscillation around $\gamma_k$  (that is, the growth of its $\theta$-component) is estimated in terms of the radial component.
This will require a careful comparison of the energies of the trajectories and their radial projections (Lemma \ref{zl2}).

\item In order to ensure that $\gamma$ remains in the region  $D[\rho_0, \theta_+]$ (Proposition \ref{l5}),  the radial direction is shown to grow fast at each oscillation, so that the amplitude of the oscillations cannot increase enough to escape the angular region (due to the angular bound in the previous item). This assures an analog to the point (a) above, and allows Lemma \ref{l_ODEincomplete} to yield incompleteness as in (b). \er
For this aim, an important role is played by formula \eqref{xx1} below. It shows that the possible oscillations of the trajectory $\gamma$  (measured by the angular length $\overline{\theta}$) will contribute exponentially to the growth of the radial coordinate.  This  allows to bound the amplitude of the oscillations (so that $\gamma$ cannot escape from $D[\rho_0, \theta_+]$) and underlies  the  incompleteness of $\gamma$.
\end{enumerate}
For the non-autonomous case, the basic idea is the same one, but there are certain technical modifications which will be analyzed in the last part, subsection  \ref{s3.4}.

\section{Autonomous and non-autonomous cases}\label{s3}

In this section, we always consider the autonomous case except in its last subsection, where the extension to the non-autonomous one will be studied.
Thus, assume that $V$ is autonomous and polynomially bounded and, so, an harmonic polynomial
of degree greater than $2$.
From the polar expression \eqref{e_vpolar}, 
\begin{equation}
\label{zevpolar}
V(\rho,\theta)=-\rho^n\cos(n\theta)-\sum_{m=1}^{n-1}\lambda_m\rho^{m}\cos m(\theta+\alpha_m),\quad n\geq 3,
\end{equation}
where  $\lambda_m$, $\alpha_m$ are real scalars for all $m=1,\ldots,n-1$ and, by Remark \ref{r_v_simpliifcation}, \br we have assumed \er
$\alpha \, (\equiv \alpha(u))=0$ (rotate $\theta=\bar\theta-
\alpha/n$) with $\lambda =1$ (use an homothety and, eventually, rotate $\theta= \bar\theta+\pi/n$), plus $\lambda_0=0$.

\subsection{Claimed incomplete trajectories}
From the previous expression for $V$, the computation in polar coordinates of its gradient in ${\mathbb R}^2$ yields:
\begin{equation}\label{ze3}
\begin{array}{rl}
-\nabla V= & -\partial_{\rho}V \; \partial_{\rho}-\partial_{\theta} V \;  \partial_{\theta}/\rho^2 \\ =
& \left( n\rho^{n-1}\cos(n\theta)+\sum_{m=1}^{n-1}m \lambda_{m}\rho^{m-1}\cos m(\theta+\alpha_m)\right)\partial_{\rho} \\  & -\left(n\rho^{n-2} \sin(n\theta)
+\sum_{m=1}^{n-1}m
\lambda_{m}\rho^{m-2}\sin m(\theta+\alpha_m)\right)\partial_{\theta}.
\end{array}
\end{equation}
Taking into account the term of greatest order in $\rho$ in the expression (\ref{ze3}), the points where
$\partial_{\theta}V$ vanish for
 $\rho$ big enough, are given by $n$ angles $\vartheta_{k}(\rho)\in[0,2\pi)$, $k=0,\ldots,n-1$, and, after reordering:
\[
\lim_{\rho\rightarrow\infty}\vartheta_{k}(\rho)= \thk:=\frac{2\pi k}{n},\;\; \qquad \qquad
k=0,\ldots,n-1.
\]
Moreover, $-\partial_\rho V (\rho,\thk )>0$ for large $\rho$ too. Then, consider the following domains of $\R^2$, for (large) $\rho_0>0$ and any $\theta_+\in
(0,\pi/(2n))$
\begin{equation}\label{e_Dk}\begin{array}{l}
D_k[\rho_0,\theta_+]:=\{z\in {\mathbb R}^2: \rho(z)>\rho_0, |\theta(z)-\thk| <
\theta_+\},\quad k=1,\ldots,n-1 \\
D[\rho_0,\theta_+]:=D_{k=0}[\rho_0,\theta_+]=\{z\in {\mathbb R}^2: \rho(z)>\rho_0, |\theta(z)| <
\theta_+\}.
\end{array}
\end{equation}
Our aim will be to prove incompleteness in \br  any \er of such domains.
\begin{theorem}\label{zt}
Let $V$ be a harmonic polynomial of degree $n\geq 3$ normalized as in \eqref{zevpolar}. For each $\theta_0, \theta_+\in (0,\pi/(2n)), \theta_0< \theta_+$,
and $k\in\{0, ...,n-1\}$ there exists some
$\rho_0>0$ such that any $V$-trajectory $\gamma\equiv (\rho,\theta)$  with initial conditions satisfying
\begin{equation}\label{e_theorem_zt}
\gamma(0)=(\rho(0),\theta(0))\in  D_k[\rho_0,\theta_0]  \qquad   \dot\rho(0) \geq 0, \;
 \dot\theta(0) = 0,
 \end{equation}
remains in $D_k[\rho_0,\theta_+]$ and are incomplete.
Therefore, a harmonic polynomial potential has complete trajectories if and only if its degree $n$ satisfies
$n\leq 2$.
\end{theorem}

\begin{remark}\label{r_simplificayrefinahomog} {\em
For the proof of Theorem \ref{zt}, only  the case
\br k=0 \er
(and, thus, the set $D[\rho_0,\theta_0]$) will be considered, with no loss of generality.

Theorem \ref{zt} yields the existence of ``many'' incomplete trajectories. This is enough for our purposes and we do not worry about a closer study, due to the rather technical character of our computations. Nevertheless, when there is no any additional complication, some of the estimates may be slightly sharper than strictly required. So, the reader could trace the proof to obtain  further results in relevant concrete examples.

In particular, the reader can trace the proof to check that, {\em in the case that $V$ is homogeneous (and autonomous), all the trajectories with $\dot\theta(0)=0$ and $\rho(0)>0$,  $\theta(0)\in (-\pi/(2n),\pi/(2n))$}
(that is, belonging to $D[\rho_0,\pi/(2n)]$ for some $\rho_0>0$) {\em are incomplete} (see Remark \ref{r_simplificayrefinahomog2}). 
}\end{remark}

From the viewpoint of differential equations, the incompleteness stated in
Theorem \ref{zt} comes from  the incompleteness to the right of the solutions for $n\geq 3$ of the system:
\begin{equation}\label{ze.14'}
\left\{\begin{array}{l} \ddot{\rho}-\rho\dot{\theta}^{2}=-\partial_\rho V(\rho,\theta) \\
\ddot{\theta}+\frac{2}{\rho}\dot{\rho}\dot{\theta}=-\frac{1}{\rho^2}\frac{\partial
V}{\partial \theta}(\rho,\theta) \\
\rho(0)>  \rho_0 (>0),\quad \dot{\rho}(0)\geq 0, \quad |\theta(0)|<
\theta_+ \in (0,\pi/(2n)),\quad \dot{\theta}(0)=0,
\end{array}\right.
\end{equation}
which satisfy additionally $|\theta(0)|<\theta_0$.
The solutions are considered to be defined on a maximal  interval $I=[0,b), b\leq \infty$, and the claimed incompleteness means $b<\infty$ (in the case $\dot\rho(0)\leq 0$, an analogous result on incompleteness to the left would hold).
Here, $\theta_0, \theta_+$ have been  prescribed and  $\rho_0$ is to be determined, so that $V$ will satisfy certain subtle inequalities, stated in the following lemma.

\begin{lemma} \label{zl1}
Choosing $\theta_+\in (0,\pi/(2n))$ and $\delta_0\in (0, \cos (n\theta_+))$, there exists some $\rho_0>1$ such that

\begin{equation}\label{ze17bis}
-\partial_\rho V(\rho,\theta) >  \delta_0  n
\rho^{n-1} (>0), \quad \quad \forall \rho\geq\rho_0, \forall
\theta\in (-\theta_+,\theta_+).
\end{equation}
Moreover, if we choose
 $\epsilon, \theta_-$ with $0<\epsilon<\theta_-<\theta_+$,
then some $\rho_0>1, \delta>0$ can be taken so that  all the following inequalities also hold  for any $\theta_1\in [\theta_-, \theta_+]$:
\begin{eqnarray}\label{ze16}
V(\rho,\theta_1)-V(\rho,\theta)> \delta  (\theta_1-\theta)
 \rho^{n} , \quad \quad \forall
\rho\geq\rho_0, \forall \theta\in (-\theta_1  +\epsilon ,
\theta_1),
\\ \label{ze16bis}
V(\rho, -\theta_1 )-V(\rho,\theta)>\delta  (\theta_1+\theta)
\rho^{n}, \quad \quad \forall
\rho\geq\rho_0, \forall \theta\in
(-\theta_1,\theta_1  -\epsilon),
\end{eqnarray}
\begin{eqnarray}\label{ze17}
 \partial_\rho V(\rho,\theta_1)-\partial_\rho V(\rho,\theta)> \delta n  (\theta_1-\theta)
\rho^{n-1},  \quad \quad \forall
\rho\geq\rho_0, \forall \theta\in
(-\theta_1  +\epsilon ,
\theta_1),
\\ \label{ze17b}
\partial_\rho V(\rho, -\theta_1  )-\partial_\rho V(\rho,\theta)> \delta n (\theta_1+\theta)
\rho^{n-1},  \quad \quad \forall
\rho\geq\rho_0, \forall \theta\in (-\theta_1 ,
\theta_1- \epsilon),
\end{eqnarray}
Finally,  chosen  $0<\theta_- < \theta_+ <\pi/(2n)$, one can find    $A>0$,  $\rho_0>1$ 
(bigger than any prescribed constant)\footnote{\label{foot p12lema4.3} Indeed, once a value $A$, $\rho_0$
has been obtained (for some prescribed $\theta_\pm$),
$\rho_0$ can be replaced by any other bigger
number maintaining the value of $A$ (and vice-versa); so, further lower bounds for $\rho_0$ \br will be incorporated later. \er  }  such that all the previous inequalities  hold for $\rho\geq \rho_0$,
by   replacing $\epsilon$ and $\delta$ by $A/\rho$ and \br $1/\rho$, \er 
resp., 
in \eqref{ze16}---\eqref{ze17b}.
\end{lemma}

\noindent {\em Proof.}
From the expression \eqref{ze3}, one has:
\begin{equation}\label{zepartialV}
\begin{array}{rl}
 -\partial_\rho V= & n\rho^{n-1}\cos(n\theta)+\sum_{m=1}^{n-1}m \lambda_{m}\rho^{m-1}\cos m(\theta+\alpha_m).
 \\
 \geq & n\rho^{n-1}\left(\cos(n\theta)-\sum_{m=1}^{n-1}\frac{m}{n}
 |\lambda_{m}|\rho^{m-n}\cos m(\theta+\alpha_m)\right).
 \end{array}
\end{equation}
As $\cos(n\theta)> \cos(n\theta_+)>0$ and
$\rho^{m-n}\leq 1/\rho$ in the large,
some big $\rho_0> 1$  so that  $\delta_0 \leq \cos (n\theta_+)-C/\rho_0$  (for, say, $nC=\sum_m m|\lambda_m|$) yield \eqref{ze17bis}.

In the remainder, we will focus on the inequalities for
$\partial V/\partial \rho$, being those for  $V$ analogous.
To obtain \eqref{ze17}, \eqref{ze17b}, recall that,  for any $\theta_1\in (0,\theta_+]$, $\theta\in (-\theta_1,\theta_1)$,  one has, by using \eqref{zepartialV}:
\begin{equation}\label{e_pu}
\begin{array}{c}
-\partial_\rho V(\rho, \theta) + \partial_\rho V(\rho, \theta_1)
\geq
\\
n\rho^{n-1}\left(\cos
(n\theta)-\cos(n\theta_1)-\sum_{m=1}^{n-1} \frac{m}{n}|\lambda_m|\rho^{m-n}|\cos
m(\theta+\alpha_m)-\cos m(\theta_1+\alpha_m)|\right).
\end{array}
\end{equation}
The expression between the big parentheses will be shown to be $\geq \delta(\theta_1-\theta)$,  
as required in \eqref{ze17},
by considering two steps, the first one when $|\theta|$ is far from $\theta_1$ and the second one when it is close.

 {\em Step 1}. To prove that, for any $\epsilon_1\in (0,\theta_-)$,  there exists $\rho_1(\epsilon_1)>1$ and $\delta_1>0$  such that the big parenthesis 
is  $\geq \pi \delta_1/n$  for $\rho\geq \rho_1(\epsilon_1)$,  whenever $|\theta| \leq \theta_1-\epsilon_1$ (so that $\pi\delta_1/n> (\theta_1-\theta)\delta_1$).  Let
\begin{equation}\label{e_ae0}
 \begin{array}{rl}
a(\epsilon_1):= &  \cos(n(\theta_- -\epsilon_1))-\cos(n\theta_-) \\
= &
  \hbox{Min}_{\theta_1\in [\theta_-,\theta_+]}\{\cos(n(\theta_1-\epsilon_1))-\cos(n\theta_1)\} \quad (>0) ,
  \end{array}
 \end{equation}
the last equality because the expression in curly brackets grows with $\theta_1$. So, putting  $nC=2\sum_m m|\lambda_m|$ one has just to ensure
$
a(\epsilon_1)-C/\rho_1(\epsilon_1)>\pi\delta_1/n
$. Clearly, it is enough to choose $\delta_1=n\, a(\epsilon_1)/(2\pi) $ and
\begin{equation}\label{e_choice_rho}
\rho_1(\epsilon_1)> \hbox{Max}\left\{1,\frac{2 C}{a(\epsilon_1)}\right\}.
\end{equation}

 {\em Step 2}. To prove that, choosing $\epsilon_1\in(0,\epsilon]$  small enough, there exists $\rho_2(\epsilon_1)$ such that   the big parenthesis of \eqref{e_pu} is $\geq \delta_2 (\theta_1-\theta)> 0$ for some $\delta_2>0$, whenever $ \theta_1-\epsilon_1\leq \theta < \theta_1$ and $\rho\geq \rho_2(\epsilon_1)$.
Indeed, as
 $|\theta|<\theta_1\leq\theta_+<\pi/(2n)$ and $0<\sin(n\theta_-)\leq \sin(n\theta_1)$, one has, by expanding in power series around $\theta_1$,
\begin{equation}\label{e_pu2}
\begin{array}{rl}
\cos(n\theta)-\cos(n\theta_1) \geq & n \sin (n\theta_-) (\theta_1-\theta) - \mu_1 (\theta_1-\theta)^2 \\
|\cos
m(\theta+\alpha_m)-\cos m(\theta_1+\alpha_m)| \leq & m 
(\theta_1-\theta) + \mu_2 (\theta_1-\theta)^2,
\end{array}
\end{equation}
for some $\mu_1,\mu_2>0$  which are chosen independent of $\theta_1$ (recall we are assuming $\theta_1\in [\theta_-,\theta_+]$).
So, it is enough to take $\rho_2(\epsilon_1)$
big enough so that $0<(n\sin n\theta_-)/2 -2C/\rho_2(\epsilon_1)$ and, then,
$0<\delta_2<(n\sin n\theta_-)/2 -2C/\rho_2(\epsilon_1))$, to obtain the required inequality.

 Thus, once such an $\epsilon_1\leq \epsilon$ has been chosen in Step 2, one can obtain
$\rho_1
(\epsilon_1)
$ from Step 1,
and \eqref{ze17} follows by taking $\rho_0=$
Max$\{\rho_1(\epsilon_1), \rho_2(\epsilon_1)\}$
and $\delta<$ Min$\{\delta_1,\delta_2\}$.

 To get \eqref{ze17b}
 repeat the process above by replacing all the bounds involving $\theta_1$ by
analogous ones with $-\theta_1$. Then, obtain  new values of  $\epsilon_1
(<\epsilon )$,   $\rho_1(\epsilon_1)$, $\rho_2(\epsilon_1)$, and choose some $\rho_0$ bigger  (resp. some $\delta$ smaller) than the values obtained in both processes. Analogously, the other bounds   \eqref{ze16}, \eqref{ze16bis} can be obtained preserving the already obtained
\br ones \er  (including \eqref{ze17bis}).

 For the last assertion, let us check just that $A, \rho_0$ can be chosen  as required for \eqref{ze17} (the other inequalities can be obtained analogously reasoning as  above).
So, our aim will be to estimate a lower bound for $\epsilon, \delta$ in terms of (big) $\rho$, taking into account the bounds in the steps 1 and 2 above.

Given $\theta_-, \theta_+$, if one chooses $\rho_1>1$ and $\epsilon_1\in (0,\theta_-)$ fulfilling  \eqref{e_choice_rho}, then
the bound \eqref{ze17} obtained in the Step 1 for $\theta\in (-\theta_1+\epsilon,\theta_1-\epsilon)$ will hold putting  $\delta_1(\epsilon_1)=n \, a(\epsilon_1)/(2 \pi)$. Recalling  the definition of  $a(\epsilon_1)$ in \eqref{e_ae0} and its series expansion around $\epsilon_1=0$, one has:
\begin{equation}\label{e_auxil1}
\frac{2 C}{a(\epsilon_1)} < \frac{4 C}{\br n   \er\sin(n\theta_-)} \, \frac{1}{\epsilon_1}, \qquad \qquad \delta_1(\epsilon_1)> \frac{\br n^2 \er \sin(n\theta_-)}{4\pi} \epsilon_1 \, ,
\end{equation}
 for small $\epsilon_1$. So, choose some big $A$ (say, $A \geq 4(\br C \er +\pi)/(n  \sin(n\theta_-))$) and put
\begin{equation}
\label{e_auxil2}
\epsilon_1(\rho)=\frac{A}{\rho} , \qquad \qquad \hat \delta_1(\rho)= \br \frac{1}{\rho} \er 
\end{equation}
Comparing (\ref{e_auxil1}) and (\ref{e_auxil2}), one finds some (arbitrarily big)  $\rho_1>1$ such that, whenever $\rho\geq \rho_1$,
  $\rho$, $\epsilon_1(\rho)$ fulfill \eqref{e_choice_rho} and, moreover, $\hat\delta_1(\rho)<\delta_1(\epsilon_1(\rho))$. Thus,
the inequalities in  the Step 1  hold for such $\epsilon_1(\rho), \hat\delta_1(\rho)$.

In order to ensure also  the bound in the Step 2 for $\theta\in [\theta_1-\epsilon_1(\rho), \theta_1)$, recall first that, once  $\epsilon_1$ is prescribed, the value of  $\rho_2(\epsilon_1)$ in that step was also valid for smaller values of $\epsilon_1$ (that is, one can take $\rho_2(\epsilon_1')=\rho_2(\epsilon_1)$  for any $0<\epsilon_1'<\epsilon_1$). So, the required bound holds just taking   $\rho_0$ as the maximum between the value of $\rho_1$ previously obtained and  $\rho_2(\epsilon_1(\rho_1))$.
 $\Box$


\begin{remark}{\em
It is worth emphasizing why $\epsilon>0$ is required in
\eqref{ze16}---\eqref{ze17b},  as this will become  an important difficulty for our problem.  It is obvious that  the non-strict inequality with $\epsilon=0$ in \eqref{ze16}, \eqref{ze16bis} will hold for $\theta=\theta_1$ and $\theta=-\theta_1$ only  when the potential satisfies $V(\rho,\theta_1)=V(\rho,-\theta_1)$.
So, the estimates \br  $\epsilon \rho_0 \approx A, \delta \rho_0 \approx 1$ \er 
can be regarded as a sort of bound for the ``asymptotic lack of symmetry'' of $V$  (originating by the terms of lower order, $m<n$, in the expression of $V$) with respect to $\theta=0$. Of course, if $V$ were symmetric ($V(\rho,\theta)=V(\rho,-\theta)$) then the half line $\theta=0$ could be reparametrized as a trajectory and its incompleteness  would \br follow \er easily as in the homogeneous case explained in Subsection \ref{s2.3}.

More quantitatively, in the case $-\theta_1<\theta<0$ we can write, instead of \eqref{e_pu2}:
\begin{equation}\label{e_pu3}
\begin{array}{rl}
\cos(n\theta)-\cos(n\theta_1) \geq & n \sin (n\theta_1) (\theta_1+\theta) - \mu_1 (\theta_1+\theta)^2 \\
|\cos
m(\theta+\alpha_m)-\cos m(\theta_1+\alpha_m)|
= & |\left(\cos(m\theta)-\cos(m\theta_1)\right)
\cos\alpha_m
\\ - &
 \left(\sin(m\theta)-\sin(m\theta_1)\right)  \sin\alpha_m|
\end{array}
\end{equation}
 So, we  cannot  take
$\rho_0$ big enough so that the right hand side of \eqref{e_pu} is positive for all $\rho \geq  \rho_0$ provided that $|\theta|<\theta_1$.
In fact, this will not hold if $\sin(\alpha_m)\neq 0$ even when  $\rho$ is very big: \br as the sinus is an odd function, \er the sum at the right hand side of \eqref{e_pu3} will be small but positive (even if divided by a big $\rho$), while the difference in the cosines may be taken arbitrarily small \br close to $\theta=-\theta_1$. \er
}\end{remark}

\begin{remark}\label{r_simplificayrefinahomog2}{\rm  In the case that the polynomial $V$ is homogeneous,  the inequalities \eqref{ze16}---\eqref{ze17b} with $\delta=0$ and $\epsilon=0$ are enough in order to prove the  (sharper) result stated in Remark \ref{r_simplificayrefinahomog}. In fact, in contrast with the non-homogeneous case explained in the previous remark, the case  $\epsilon=0$ can be proved  easily in the homogeneous case, as lower order terms of the polynomial do not exist.
}\end{remark}

\subsection{Balance of energies and amplitude of the oscillations}
Next, the conservation law of energy \eqref{e18} will be  used to estimate the amplitude of the oscillations of the $\theta$-coordinate.
 More precisely, particularizing  \eqref{e18}  for a trajectory $\br (\rho(s),\er \theta(s))$ in  the autonomous case,
\begin{equation}\label{ze18_aut}
\frac{1}{2}(\dot{\rho}(s)^2+\rho(s)^2\dot{\theta}(s)^2)+ V(\rho(s),\theta(s)) =E(s_0)
\qquad
\forall s\in I=[0,b),
\end{equation}
 where $E(s_0)$ is the total energy \br at some $s_0\in I$. \er
We will assume always   $\dot \theta({s_0})=0$ (in particular, this happens if $s_0=0$ by \eqref{ze.14'}) and, thus,
$E(s_0)=\dot\rho(s_0)^2/2 + V(\rho(s_0),\theta(s_0)).$

For any trajectory  $\gamma(s)\equiv (\rho(s), \theta(s))$,
$s\in [0,b)$,
of \eqref{ze.14'},  and $\theta_1\in \R$,   the curve $\gamma_{\theta_1}(s)\equiv (\rho(s), \theta_1)$ will be called
 its {\em $\theta_1$-projection}, and it will have an energy equal to the radial one of $\gamma$ plus the potential one corresponding to $V(\rho(s), \theta_1)$. In particular, choosing $s_1\in [0,b)$,
the energy of the $\theta(s_1)-$projection is:

\begin{equation}\label{ef1}
F(s):=\frac{1}{2}\dot{\rho}(s)^2+V(\rho(s),\theta(s_1))
\end{equation}
\br The next result provides a first estimate on the (angular) oscillations of the trajectory. The key in the proof will use that, whenever $\theta(s)$ is monotonous, the energy $F$ cannot decrease. \er

\begin{lemma}\label{zl2} Choose $0<\theta_-<\theta_+ <\pi/(2n)$, and let $A, \rho_0>0$ be as in the last part of Lemma \ref{zl1}.
Let $(\rho(s)$, $\theta(s))$, $s\in [0,b)$, be a solution of \eqref{ze.14'} with $(\rho(s_0), \theta(s_0))\in D[\rho_0,\theta_+]$, $\dot{\theta}(s_0)=0$ and $\dot\rho(s_0)\geq 0$ for some $s_0\in [0,b)$.

If $s_1\in (s_0,b)$ satisfies $|\theta(s_0)|<|\theta(s_1)|<\theta_+$, and $\theta(s)$ is (non-necessarily strictly) monotonous on $(s_0,s_1)$, then   $|\theta(s_1)|-|\theta(s_0)|\leq A/\rho(s_0)$.  
\end{lemma}
\noindent {\em Proof}.
By using the first equation of (\ref{ze.14'}), and the inequality (\ref{ze17bis}), we have
\begin{equation}\label{zyy}
\ddot{\rho}(s)\geq -\partial_\rho V(\rho(s),\theta(s))>\delta_0 n\rho(s)^{n-1}>0, \quad \forall s\in [s_0,s_1]
\end{equation}
for some $\delta_0>0$ and, thus,
\begin{equation}
\label{zedr}
\dot \rho(s)>0 \quad \quad \forall s\in(s_0,s_1].
\end{equation}
In particular, $\rho$ is strictly increasing in $[s_0,s_1]$.
On the other hand, adding $V(\rho(s),\theta(s_1))$ in the conservation law (\ref{ze18_aut}), the energy of the $\theta(s_1)$-projection \eqref{ef1} satisfies
\begin{equation}\label{zg}
\begin{array}{c}
F(s) =
E(s_0)-\frac{1}{2}\rho(s)^2\dot{\theta}(s)^2+ V(\rho(s),\theta(s_1))-V(\rho(s),\theta(s)).
\end{array}
\end{equation}
Next, assume by contradiction that
\begin{equation}\label{et1}
|\theta(s_1)|-|\theta(s_0)|>\frac{A}{\rho(s_0)}.
\end{equation}
Putting $\theta_1= |\theta(s_1)|$,  $\theta= \theta(s_0)$, $\epsilon =\frac{A}{\rho(s_0)}$,  one has
\begin{equation}
\label{e_lastp} \theta \in (-\theta_1+\epsilon, \theta_1-\epsilon), \end{equation}
and  the last part of Lemma \ref{zl1} will be applicable. Then,
\begin{equation}\label{zvv}
F(s_1)-F(s_0)\leq -(V(\rho(s_0),\theta(s_1))-V(\rho(s_0),\theta(s_0)))<0,
\end{equation}
the last inequality by  \eqref{ze16}, \eqref{ze16bis}. By the mean value theorem applied to $F$ on $[s_0,s_1]$, one deduces the existence of $s_{*}\in (s_0,s_1)$ such that $\dot F(s_*)<0$ (and $|\theta(s_*)| \leq |\theta(s_1)|<\theta_+$). That is, from \eqref{ef1}:
\begin{equation}\label{zs}
\ddot{\rho}(s_*)\dot{\rho}(s_*)<-\dot{\rho}(s_*)\partial_{\rho}V(\rho(s_*),\theta(s_1)).
\end{equation}
Moreover, from \eqref{zedr}
this simplifies into:
\begin{equation}\label{zeq}
\ddot{\rho}(s_*)<-\partial_{\rho}V(\rho(s_*),\theta(s_1)).
\end{equation}
Now the assumed  monotonicity of $\theta(s)$ and inequality $|\theta(s_0)|\leq |\theta(s_1)|$ yield two possibilities: either $|\theta(s_*)|\leq |\theta(s_0)|$ or $|\theta(s_0)|<|\theta(s_*)|<|\theta(s_1)|$. In any of them, the inequalities assured by the
last paragraph of Lemma \ref{zl1} become applicable  with $\theta= \theta(s_*)$, $\theta_1= |\theta(s_1)|$.
 Indeed, in the first case \br the inclusion \er \eqref{e_lastp} holds directly putting now $\theta=\theta(s_*)$. In the second one, if $\theta(s_1)>0$ (resp. $\theta(s_1)<0$) then $\theta(s)$ is non-decreasing (resp. non-increasing) and, then $\theta(s_*)>0$ (resp. $\theta(s_*)<0$), making  \eqref{ze17} (resp. \eqref{ze17b}) applicable. Summing up, we deduce from \eqref{ze17} and \eqref{ze17b}
\begin{equation}\label{eee}
\partial_{\rho} V(\rho(s_*),\theta(s_1))-\partial_{\rho} V(\rho(s_*),\theta(s_*))>0.
\end{equation}
\newline Nevertheless, 
inequalities (\ref{eee}) and (\ref{zeq}) imply
\begin{equation} \label{e20}
\begin{array}{rl}
\ddot{\rho}(s_*)<-\partial_{\rho} V(\rho(s_*),\theta(s_*)),
\end{array}
\end{equation}
in contradiction with the first inequality in (\ref{zyy}). $ \Box$

\subsection{Proof of the main result}
We will start by proving the following part of Theorem \ref{zt}.
\begin{proposition}\label{l5} Let $V$ be as in Theorem \ref{zt} with $n\geq 3$. For any $0<\theta_0< \theta_+ < \pi/(2n)$, there exists $\rho_0>1$
such that any solution
 $\gamma(s)=(\rho(s), \theta(s)$), $s\in [0,b)$ of \eqref{ze.14'} satisfying
 \begin{equation}
 \label{e_condicionesiniciales}
 |\theta(0)|<\theta_0,\quad(\hbox{plus} \quad \rho(0)> \rho_0 , \quad  \dot \rho(0)\geq 0 , \quad \dot{\theta}(0)=0)
 \end{equation}
remains in $D[\rho_0,\theta_+]$ for all $s\in [0,b)$.
 \end{proposition}
With this aim, two technical bounds are provided before. The first one becomes straight\-forward from the expression of $\partial_\theta V$ in \eqref{zevpolar}, \eqref{ze3} (recall  the proof of Lemma~\ref{zl1}).
\begin{lemma} \label{l_previo} For any $\bar\delta> 1$ and $\theta_0\in (0,\pi/2)$ there exists $\rho_0>1$ such that
\[
\begin{array}{rl}
|\partial_\theta V(\rho,\theta)|\leq
\overline{\delta} n\rho(s)^{n} & \hbox{for all $\rho\geq\rho_0$, $\theta\in \R$}\\
\partial_\theta V(\rho,\theta)>0 & \hbox{for all $\rho\geq\rho_0$, $\theta\in [\theta_0,\pi/(2n)]$}
\\
\partial_\theta V(\rho,\theta)<0 & \hbox{for all $\rho\geq\rho_0$, $\theta\in [-\pi/(2n),-\theta_0]$.}
\end{array}
\]
\end{lemma}
For the second bound, let us denote by
\begin{equation}\label{aq4}
\overline{\theta}(s):=\int_{s_0}^s|\dot{\theta}(\sigma)|d\sigma,\quad s\in [s_0,b)
\end{equation}
the \br {\em angular length} \er  covered by such a trajectory starting at some suitable $s_0$ in its domain $[0,b)$.
Note that $\bar\theta$ is non-decreasing and it satisfies $\dot{\bar\theta}(s)= |\dot\theta(s)|$. \br It is also
$C^2$ except when $\dot\theta(s)=0$ and $\ddot\theta(s)\neq 0$ (thus, only at isolated points);
anyway, $\dot{\bar\theta}$ has one-sided derivatives elsewhere.
So, one can write $\dot{\bar\theta}= \int \ddot{\bar\theta}(s)ds$  (in fact, the zeroes
of  $\dot \theta(s)$  will not play any relevant role for $\dot{\bar\theta}$).
However, the  bounds for $\ddot{\bar\theta}(s)$ to be obtained below can be regarded as bounds for both one-sided derivatives of $\dot{\bar\theta}$ at each point. \er
Now, our aim is to bound the growth  of the angular length $\overline{\theta}$ by the growth of $\rho$. \br The following technical lemma will allow to bound the amplitude of consecutive oscillations (as a difference with other results, it will be applied only when $s_0>0$ and, thus, the assumption $\dot\rho(s_0)>0$ will not be restrictive.) \er

\begin{lemma}\label{l5bis}  For some $\rho_0$  big enough, any solution $\gamma: [s_0,b)\rightarrow \R^2$ (for $n\geq 3$)
starting at $D[\rho_0,\theta_+]$
with      $\dot\rho(s_0) > 0, \dot\theta(s_0)=0$
satisfies:
\begin{equation}\label{xx1}
\rho(s)> \rho(s_0)e^{\overline{\theta}(s)/\Lambda}
\qquad  \hbox{for all} \; s\in(s_0,b) \; \hbox{such that } \; \gamma([s_0,s])\subset D[\rho_0,\theta_+],
\end{equation}
where $\Lambda:=8/\cos(n\theta_+)$ ($> 1$).
\end{lemma}

\noindent {\em Proof}. From the first equation \eqref{ze.14'} and Lemma \ref{zl1} for, say,  $\delta_0=\cos(n\theta_+)/2$, there exists $\rho_0>1$  such that the following improvement of \eqref{zyy}   holds:
\begin{equation}\label{aq1}
\begin{array}{rl}
\ddot{\rho}(s)= & \rho(s)\dot{\theta}(s)^2-\partial_\rho V(\rho(s),\theta(s))  \\ >
& \rho(s)\dot{\theta}(s)^2 + \delta_0 n\rho(s)^{n-1}  = \rho(s)(\dot{\bar{\theta}}(s)^2 + \delta_0 n\rho(s)^{n-2})
\end{array}
\end{equation}
(whenever remaining in $D[\rho_0,\theta_+]$).
\br By using $\dot\rho(s)>0$
(see \eqref{zedr}) and the second equation in \eqref{ze.14'}, the  inequalities
\begin{equation}\label{aq2}
\ddot{\overline{\theta}}(s)\leq
|\partial_\theta V(\rho(s),\theta(s))|/\rho(s)^2\leq
\overline{\delta} n\rho(s)^{n-2}
\end{equation}
hold for, say,  $\overline{\delta}=2$. Indeed, for the first inequality, use that,  when $\dot\theta(s)>0$ (resp. $\dot\theta(s)<0$),  $\ddot{\bar\theta}=\ddot{\theta}$ (resp. $\ddot{\bar\theta}=-\ddot{\theta}$); then, when  $\dot\theta(s)=0$, either $\ddot{\overline{\theta}}(s)=0$ or the inequality holds for the lateral derivatives of $\dot{\overline{\theta}}(s)$. The second inequality is straightforward from Lemma \ref{l_previo}. \er

The choice of $\Lambda$ in the statement of the lemma implies $\overline{\delta}/\Lambda<\delta_0$, so, put $c_0:=(\delta_0-\overline{\delta}/\Lambda)n\rho_{0}^{n-2}$ and
\[
h(s):={\rm max}\left\{\frac{\dot{\overline{\theta}}(s)^2}{\Lambda^2}+\frac{\ddot{\overline{\theta}}(s)}{\Lambda},c_0\right\} (>0), \qquad \qquad \forall s\in [0,b).
\]
As $\delta_0 n\rho(s)^{n-2}>c_0$, inequalities (\ref{aq1}) and \eqref{aq2} yield directly
\begin{equation}\label{aq3}
\ddot{\rho}(s)\geq \rho(s)(\dot{\overline{\theta}}(s)^2+\delta_0 n\rho(s)^{n-2})>\rho(s) h(s),\quad\forall s\in [s_0,b),\end{equation}
whenever in  $D[\rho_0,\theta_+]$.
\br Formula  \eqref{xx1} follows by comparing (\ref{aq3}) with the function
\[
f(s)=\rho(s_0)e^{\overline{\theta}(s)/\Lambda} ,
\]
which has the same regularity as $\bar\theta$ and it is a solution of the
ODE
\begin{equation}\label{aq5}
\ddot{f}(s)=f(s)\left(\frac{\dot{\overline{\theta}}(s)^2}{\Lambda^2}+\frac{\ddot{\overline{\theta}}(s)}{\Lambda}\right),\quad f(s_0)=\rho(s_0),\;\; \dot{f}(s_0)=0.
\end{equation}
\er Indeed for $s>s_0$ close to $s_0$, the strict inequality (\ref{xx1}) (i.e.,
$\rho(s)>f(s)$), holds just because $\dot\rho(s_0)> \dot f(s_0) (=0)$. So, assume by contradiction that $s_1>s_0$ is the first point where $\rho(s_1)=f(s_1)$ and, thus,
$\dot{\rho}(s_1)\leq \dot{f}(s_1)$.
\br Nevertheless,
$$\dot{\rho}(s_1)-\dot{f}(s_1)= \int_{s_0}^{s_1}(\ddot \rho-\ddot f)(s)ds + \left(\dot{\rho}(s_0)- \dot{f}(s_0)\right)>0$$
the last equality as  $\ddot{\rho}(s)> \ddot{f}(s)$ because of (\ref{aq3}) and (\ref{aq5}) (as still $\rho(s)>f(s)$ for $s\in (s_0,s_1)$), a contradiction. \er
$\Box$

\bigskip

\noindent {\em Proof of Proposition \ref{l5}}. Define $\theta_*$ and choose $\theta_-$ so that
\begin{equation}\label{e_f1}
\theta_*-\theta_0 = ( \theta_+ - \theta_0 ) /3,
\qquad
0<\theta_0<\theta_*<\theta_-<\theta_+  \; \; (<\pi/(2n)).
\end{equation}
Take $A>0, \rho_0>1$ as in the last part of Lemma \ref{zl1} and Lemma \ref{zl2} (for the prescribed $\theta_+, \theta_-$); moreover (recall footnote \ref{foot p12lema4.3}), take
$\rho_0$ also satisfying Lemma~\ref{l_previo} (for the prescribed $\theta_0$), Lemma \ref{l5bis} and the following convenient bound:
\begin{equation}
\label{e_f2}
\frac{S}{\rho_0}<\theta_*-\theta_0 ,
\;\; \hbox{for} \;
S:= A \sum_{k=0}^\infty
e^{-(\theta_*-\theta_0)k/\Lambda}
\;\; \hbox{(in particular,}
\;\;
\frac{A}{\rho_0}< \theta_*-\theta_0 \, \hbox{)}
\end{equation}
 with $\Lambda>0$ as in Lemma \ref{l5bis}
(notice $A<S<\infty$).
 A first property of any trajectory fulfilling \eqref{e_condicionesiniciales} is the following.

\bigskip
\noindent {\em Claim 0}. If $\theta(s_0)\in [\theta_0,\theta_+]$ (resp. $\theta(s_0)\in [-\theta_+,-\theta_0]$) and
$\dot\theta(s_0)=0$, then $\ddot\theta(s_0)<0$ (resp. $>0$). Moreover,  whenever $\gamma(s)$ does not enter in $D[\rho_0, \theta_0]$, necessarily $\dot\theta(s)<0$ (resp. $>0$) and $\theta$ decreases (resp. increases) beyond  $s_0$.
In particular, all the zeroes of $\dot\theta$ with $\theta$ in $[\theta_0, \theta_+]$ are strict relative maxima (resp. minima) and, thus, isolated.

\smallskip

{\em Proof of Claim 0.} The inequality for $\ddot\theta(s_0)$ follows directly from the trajectory equations (the second one in \eqref{ze.14'})  plus Lemma \ref{l_previo}.
 So, just after $s_0$ one has $\dot\theta(s)<0$ (resp. $>0$) and, if $\dot\theta(s)$ vanished at some $s'>s_0$ (before the curve has arrived at $D[\rho_0, \theta_0]$), we would obtain again the same sign for $\ddot\theta(s')$ (that is, $s'$ is a strict critical point), which is a contradiction. $\Box$

\bigskip

 Consider the increasing sequence $\{\bar s_m: m=1,\dots \}$ of the zeroes  in Claim~0
and take the  subsequence $\{s_k\}_k$  obtained by imposing: $\gamma([0,s_m])\subset D[\rho_0,\theta_+]$ for all $m$,
 $s_1$ is the first $\bar s_m$ such that  $|\theta(s_1)|>\theta_*$ and, inductively, $s_{k+1}$ is the first $\bar s_m$
 such that $|\theta(s_{k+1})|>|\theta(s_k)|$. Now, consider the following cases.

\bigskip
\noindent {\em Claim 1}. If either there is no such a $s_k$, or  there are infinitely many, $\gamma$ remains in $D[\rho_0,\theta_+]$, as required.
\smallskip

{\em Proof of Claim 1.} In the first case,
if $|\theta(s)|\leq \theta_*$ for all $s$, we are done. Otherwise, there exists some $s_*$ such that $|\theta(s_*)|=\theta_*$ and $\gamma$ escapes $D[\rho_0,\theta_*]$; moreover, $|\theta(s)|$ grows
strictly  beyond $s_*$ (as $\dot\theta(s)$ cannot vanish outside $D[\rho_0,\theta_*]$ by Claim 0). Let $s_0$ be the last critical point of $\dot\theta$ in $[0,s_*]$ (eventually, $s_0=0$). For any $s\in (s_*,b)$,
\begin{equation}\label{e_claim1}
(0<)\, \; |\theta(s)|-|\theta(s_*)| \leq  |\theta(s)|-|\theta(s_0)|\leq A/\rho(s_0)\leq A/\rho_0 < \theta_*-\theta_0
\end{equation}
(use $|\theta(s_0)|\leq |\theta(s_*)|=\theta_*\leq |\theta(s)|$, Lemma~\ref{zl2}, the fact that $\rho$ is increasing \eqref{zedr}, and the bound
\eqref{e_f2}).
So, $|\theta(s)|\leq 2\theta_*-\theta_0<\theta_+$
and $\gamma$ remains in $D[\rho_0,\theta_+]$, as required. In the case that there are infinitely many $s_k$, necessarily $\{s_k\}\nearrow b^* \leq b$. Assuming $b^* < b$ (otherwise we are done), by continuity, $\theta(b_*)\in [\theta_*,\theta_+]$ and $\dot \theta(b_*)=0$, in contradiction with Claim 0. $\Box$
\bigskip

So, $\{s_k\}_k$ will be assumed to be finite. The growth of $\theta$ between consecutive  $s_k \in [0,b)$ is estimated next.

\bigskip

\noindent {\em Claim 2.} The first term $s_1$ of the 
sequence  $\{s_k\}_k$ satisfies:
\begin{equation}
\label{e_claim2} |\theta(s_1)|-|\theta(0)|
\leq A/\rho_0
<\theta_*-\theta_0.
\end{equation}

\smallskip

{\em Proof of Claim 2}. Just notice that a bound as in
\eqref{e_claim1} holds until $s=s_1$.  $\Box$

\bigskip

\noindent {\em Claim 3.} For $s_k , s_{k+1}$ as above:
$ | \theta(s_{k+1})|-|\theta(s_k)| \leq A/\rho(s_k)
$

\smallskip

{\em Proof of Claim 3}. By construction $\theta_* < |\theta(s_k)|<|\theta(s_{k+1})|$ and we have two relevant cases (the others are analogous): (a) $\theta_* < \theta(s_k)<\theta(s_{k+1})$, and (b)~
$\theta(s_{k+1})<-\theta(s_k)$ with $\theta_* <
\theta(s_k)$. In the case (a), as $s_k, s_{k+1}$
are local maxima (by Claim 0), there exists a last
point $s_0\in (s_k, s_{k+1})$ with $\dot\theta (s_0)=0$. Thus,
 $\theta$ is monotonous in $[s_0,s_{k+1}]$ and Lemma \ref{zl2} yields $|\theta(s_{k+1})|-|\theta(s_0)|  \leq A/\rho(s_0) ( \leq A/\rho(s_k))$. So, the required bound follows
because $|\theta(s_0)|\leq |\theta(s_k)|$
(otherwise, $s_{k+1}\leq s_0$), that is,
$|\theta(s_{k+1})|-|\theta(s_k)| \leq  |\theta(s_{k+1})|-|\theta(s_0)|$. For the case (b),  either $\theta(s)$ is monotonic in $[s_k,s_{k+1}]$ and the result follows directly from Lemma \ref{zl2}, or it has a last critical point $s_0\in (s_k, s_{k+1})$, and the result follows as in the previous case. $\Box$

\bigskip

\noindent {\em Claim 4.} Let $s_{k_{+}}, k_{+}\in \NN$, be the last term of $\{s_k\}_k$. For any $s\in (s_{k_{+}}, b)$:
\begin{equation}
\label{e_claim4} |\theta(s)|-|\theta(s_{k_{+}})| \leq A/\rho(s_{k_{+}} ) \; <\theta_*-\theta_0.
\end{equation}

\smallskip

{\em Proof of Claim 4}. Assume  $|\theta(s)|>|\theta(s_{k_{+}})|$ (otherwise, it is trivial). Reasoning as
\br above \er there is a last point  $s_0\in [s_{k_{+}}, s)$ such that  $\dot\theta(s_0)=0$ \br (and $|\theta(s_0)|\leq |\theta(s_{k_+})|$). \er Then, $\theta$ is monotonous in $[s_0,s]$ and Lemma \ref{zl2} gives the bound.  $\Box$

\bigskip

In order to estimate the total increase of $\theta$ along consecutive pairs $s_k, s_{k+1}$ in Claim 3, the angular length in Lemma \ref{l5bis} will be used. From Claim 0, when $\gamma$ moves  from $s_k$ to $s_{k+1}$ one has $\overline{\theta}(s_{k+1})-\overline{\theta}(s_k)\geq \theta_*-\theta_0$. Therefore, inductively
\[
\mu k\leq \overline{\theta}(s_{k+1}),\quad\hbox{ $\forall k\in \{1,\dots, {k_{+}}-1\}$,}
\qquad\hbox{where $\mu:=\theta_*-\theta_0 \; (>0)$.}
\]
Thus, for $\Lambda>0$ as in Lemma \ref{l5bis}, and $s_1$ the first term in the sequence $\{s_k\}_k$,
$$
 \rho(s_1) e^{(\mu k)/\Lambda}\leq  \rho(s_1)e^{\overline{\theta}(s_{k+1})/\Lambda}
\leq \rho(s_{k+1}),$$
the last inequality by (\ref{xx1}) applied to $s_0=s_1$ (recall that $s_1>0$, thus $\dot\rho(s_1)>0$, and Lemma \ref{l5bis} is appliable).
That is, $$\frac{A}{\rho(s_{k+1})}\leq \frac{A}{\rho(s_1)} e^{-(\mu k)/\Lambda} < \frac{A}{\rho_0} e^{-(\mu k)/\Lambda}, \qquad \forall k\in\{0,1,\dots,{k_{+}}-1\},$$
and taking into account Claim 3 plus the expression of $S$ in \eqref{e_f2}
\begin{equation} \label{e_last}
 |\theta(s_{k_{+}})|-|\theta(s_1)| = \sum_{k=1}^{k_{+} -1}  (| \theta(s_{k+1})|-|\theta(s_k)|) \leq \sum_{k=1}^{k_{+} -1} \frac{A}{\rho(s_k)}<\frac{S}{\rho_0}<\theta_*-\theta_0.
\end{equation}
So, adding  \eqref{e_claim2}, \eqref{e_claim4}, \eqref{e_last} and using \eqref{e_f1},
$$
|\theta(s)| < \br |\theta(0)|+(\theta_+-\theta_0)
= \er
\theta_+-(\theta_0-
|\theta(0)|)
<\theta_+,
$$
and $\gamma$ must remain in  $D[\rho_0,\theta_+]$, as required. $ \Box$

\bigskip

\noindent {\em Proof of Theorem \ref{zt}.}
Taking into account Remark \ref{r_simplificayrefinahomog}, focus on $k=0$. Choose  $\theta_+\in (\theta_0,\pi/(2n))$ and take $\rho_0$ as in Proposition \ref{l5}. Then $\gamma$ remains in $D[\rho_0,\theta_+]$ and, moreover, the bound \eqref{ze17bis} in Lemma \ref{zl1} holds for some $\delta_0>0$. So, from the first equation in (\ref{ze.14'}) we have (as in \eqref{zyy}):
\begin{equation}\label{ff}
\ddot{\rho}(s)=\rho(s)\dot{\theta}(s)^{2}-\frac{\partial
V}{\partial \rho}(\rho(s),\theta(s))>\delta_0 n\rho(s)^{n-1},
\end{equation}
which is incomplete because, by hypothesis, $n\geq 3$ and, so, Lemma \ref{l_ODEincomplete} (i) applies. $\Box$

\begin{remark}{\em  The proof of Proposition \ref{l5} provides further information about the trajectories. For example,  the obtained trajectories in $D[\rho_0,\theta_+]$ may present infinite oscillations and its amplitude can be also bounded by using formula \eqref{e_last} (with infinite terms).
}\end{remark}



\subsection{The non-autonomous case} \label{s3.4}
Next, let us consider the case  $V\equiv V(z,u)$.
Taking into account Proposition \ref{p_polinomio_no_autonomo},
we can choose  $u_0\in \R, \kappa_0>0$ such that the polar
expression of the potential  \eqref{e_vpolar} is valid in
$(u_0-\kappa_0, u_0+\kappa_0)$ with $C^1$ functions and $
\lambda(u_0)\neq 0$. Applying Remark \ref{r_v_simpliifcation} as in the autonomous case, we also assume
\begin{equation}\label{e_u0_normailized}
\alpha(u_0)=0, \qquad \lambda(u_0)>1, \qquad \lambda(u)\geq 1, \quad \forall u\in I= (u_0-\kappa_0, u_0+\kappa_0).\end{equation}
The expression of the gradient $\nabla_zV$ (with respect only to  $z$) is completely analogous to \eqref{ze3}, up to the dependence of $\lambda, \lambda_m, \alpha_m$ with $u$ and the appearance of $\alpha(u)$. Additionally, one has,
\begin{equation}\label{fff}
\begin{array}{c}
\partial_u V=[-\lambda'(u)\cos n(\theta + \alpha(u))+n\lambda(u)\sin n(\theta + \alpha(u))\alpha'(u)]\rho^n\qquad\qquad\qquad
\\ \qquad -\sum_{m=1}^{n-1}[\lambda'_m(u)\cos m(\theta+\alpha_{m}(u))-\lambda_m(u)m\sin m(\theta+\alpha_{m}(u))\alpha'_m(u)]\rho^m.
\end{array}
\end{equation}
Reasoning as in the autonomous case,  the points where $-\nabla_zV$ points \br out \er the infinity radial direction are determined by $n$ angles $\vartheta_{k}(\rho,u)$ (which cancel the coefficient of
$\partial_\theta$ and make positive the one of $\partial_\rho$) for large $\rho$ satisfying:
\[
\vartheta_{k}(\rho,u)\rightarrow \thk(u):=\frac{2\pi k-\alpha(u)}{n},\;\;
k=0,\ldots,n-1,\qquad \hbox{if $\rho\rightarrow\infty$.}
\]
Stressing the comparison with Theorem
\ref{zt}, now our aim is to prove the following:
\begin{theorem}\label{t}
Let $V(z,u)$ be a non-autonomous $z$-harmonic potential which is polynomially but not quadratically polynomially $u$-bounded, and choose any $u_0\in\R$, $\kappa_0>0$ such that the polar expression of the potential \eqref{e_vpolar} holds for $I=(u_0-\kappa_0, u_0+\kappa_0)$ with leading term of degree $n\geq 3$ and normalization \eqref{e_u0_normailized}.

For each
 $\theta_0, \theta_+\in (0,\pi/(2n)), \theta_0< \theta_+$
and $k\in\{0, ...,n-1\}$ there exists some
$\rho_0>0$ such that any $V$-trajectory $\gamma: [u_0,b)\rightarrow \R^2$ inextensible to the right
with initial conditions at $s=u_0$ as those in Theorem \ref{zt} for $s=0$ (formula \eqref{e_theorem_zt}) and, additionally,
\begin{equation}\label{condicionadicional}
\dot{\rho}(\br u_0 \er) \bb > 0 \eb
\end{equation}
 remains in $D_k[\rho_0,\theta_+]$ and, moreover,  it is incomplete (being $b< u_0+\kappa_0$).

Therefore, a $z$-harmonic polynomial potential which is polynomially $u$-bounded
has complete trajectories if and only if it is a $u$-dependent polynomial in $z$ of  degree at most 2.
\end{theorem}

\begin{remark}\label{dificultadnoautonomo}{\em
(1) The density of the set $\mathcal{D}\ni u_0$ proven in Proposition \ref{p_polinomio_no_autonomo} makes apparent that, as in the autonomous case, \br   a  multiplicity of \er incomplete trajectories will be obtained.

(2) Recall that \eqref{condicionadicional} is just the strengthening of the non-strict inequality assumed in formula \eqref{e_theorem_zt} into a strict one. This additional hypothesis stresses the unique step in the non-autonomous
case with no analog in the autonomous one, namely, the Claim at the proof of
 Lemma~\ref{l2}. 
 A more careful analysis would allow one to consider this and other cases. Indeed, Theorems \ref{zt} and \ref{t} would hold not only in the case $\dot \theta(u_0)=0$ but also for values of $\dot \theta(u_0)$ small enough for the required technical bounds. Then, the case $\dot\rho(u_0)=0$ would also hold because the value of $\rho_0$ will be so big that $-\partial_\rho V>0$ on all $D_k[\rho_0, \theta_+]$ and, therefore, from the equality $\dot\rho(u_0)=0$ one obtains $\dot\rho(\bar u_0)>0$ for $\bar u_0>u_0$ arbitrarily close to $u_0$ and, so, with $\dot \theta(\bar u_0)$ arbitrarily  small. However, it is not our purpose here to yield sharp estimates of all the possible incomplete trajectories, but only to show the existence of some of them.
\eb
}\end{remark}
As in the autonomous case, we will focus only on the case  $k=0$.
Thus, the system of ODE's to be considered is similar to \eqref{ze.14'}, namely, for $n\geq 3$:
\begin{equation}\label{e.14'}
\left\{\begin{array}{l} \ddot{\rho}(s)-\rho(s)\dot{\theta}^{2}(s)=-\partial_\rho V(\rho(s),\theta(s),s) \\
\ddot{\theta}(s)+\frac{2}{\rho(s)}\dot{\rho}(s)\dot{\theta}(s)=-\frac{1}{\rho(s)^2} \partial_{\theta}V(\rho(s),\theta(s),s)
\\ 
\rho(\br u_0 \er )> \rho_0(>0),\quad\dot{\rho}(\br u_0 \er )\geq 0,\quad |\theta(\br u_0 \er )|<
\theta_+(\in (0,\pi/(2n))),\quad\dot{\theta}(\br u_0 \er )=0
\end{array}\right.
\end{equation}
\br (the reader can consider, with no loss of generality, $u_0=0$). \er Even though the domains $D[\rho_0,\theta_+]$
to be used are equal to those in \eqref{e_Dk} for the autonomous case, a technical strengthening of the bounds in Lemma \ref{zl1} is required.

\begin{lemma} \label{zl1_z}
Chosen $\theta_+\in (0,\pi/(2n))$ and $\delta_0\in (0,\cos(n\theta_+))$, there exists some $\rho_0>1$ and $\kappa\in (0,\kappa_0)$  such that

\begin{equation}\label{ze17bis_z}
-\partial_\rho V(\rho,\theta,u) > \delta_0 n
\rho^{n-1} (>0), \quad \quad \forall \rho\geq\rho_0, \, \forall
\theta\in (-\theta_+,\theta_+), \, \forall u\in [u_0,u_0+\kappa).
\end{equation}
Moreover, if we choose $\epsilon, \theta_-$ with $0<\epsilon<\theta_-<\theta_+$,
then some $\rho_0>1, \delta>0$, $\kappa>0$ can be taken so that  all the following inequalities also hold for all $\rho\geq\rho_0$,  $u\in [u_0,u_0+\kappa)$ and $\theta$ varying as specified below, for any choice $\theta_1\in [\theta_-, \theta_+]$:
\begin{eqnarray} \label{ze16_z}
V(\rho,\theta_1,u)-V(\rho,\theta,u)> \delta  (\theta_1-\theta)
 \rho^{n} ,
& 
 \forall \theta\in (-\theta_1
 +\epsilon ,
\theta_1),
\\ \label{ze16bis_z}
V(\rho, -\theta_1,u)-V(\rho,\theta,u)>\delta  (\theta_1+\theta)
\rho^{n},
& 
\forall \theta\in
(-\theta_1,\theta_1 -\epsilon),
\end{eqnarray}
\begin{eqnarray}\label{ze17_z}
\partial_\rho V(\rho,\theta_1,u)-\partial_\rho V(\rho,\theta,u)> \delta n (\theta_1-\theta)
\rho^{n-1},
& 
\forall \theta\in
(-\theta_1 +\epsilon,
\theta_1),
\\ \label{ze17b_z}
\partial_\rho V(\rho, -\theta_1,u)-\partial_\rho V(\rho,\theta,u)> \delta n (\theta_1+\theta)
\rho^{n-1},
& 
\forall \theta\in (-\theta_1 ,
\theta_1- \epsilon),
\end{eqnarray}
In addition, taking if necessary a bigger $\rho_0$, some $B>0$ can be chosen such that
\begin{eqnarray} \label{d}
\rho^{-3/2}\left|\partial
_u V(\rho, \theta_1,u)-\partial_u V(\rho,\theta,u)\right| < B (\theta_1-\theta) \rho^{n-3/2}
& \\ \label{dbis}
< \sqrt{\delta_0} \, (\partial_\rho (V(\rho, \theta_1,u)-\partial_\rho V(\rho,\theta,u)) ,
& 
\forall \theta\in
(-\theta_1 +\epsilon,
\theta_1)
\\
\label{dd}
 \rho^{-3/2}\left|\partial_u V(\rho, -\theta_1,u)-\partial_u V(\rho,\theta,u)\right| < B (\theta_1+\theta) \rho^{n-3/2}
 & \\
 \label{ddbis}
 < \sqrt{\delta_0} (\partial_\rho V(\rho, -\theta_1,u)-\partial_\rho V(\rho,\theta,u)),
& 
\forall \theta\in
(-\theta_1,\theta_1 -\epsilon)
\end{eqnarray}
Finally, chosen $0<\theta_- < \theta_+ <\pi/(2n)$, one can find $A>0$, $\rho_0>1$, $\kappa\in (0,\kappa_0)$  
such that all the previous inequalities  hold for $\rho\geq \rho_0$,
by replacing $\epsilon$ and $\delta$ by $A/\rho$ and \br $1/\rho$, \er resp., 
in \eqref{ze16_z}---\eqref{ddbis}.
\end{lemma}

\noindent {\it Proof.} Along the proof, we will always assume that  $\kappa\in(0,\kappa_0)$   is small enough so that  $\theta_+ + \alpha[\kappa]<\pi/(2n)$, where
$$\alpha[\kappa]:= \; \hbox{Max}\{|\alpha(u)|: u\in [u_0,u_0+\kappa]\}, \qquad \hbox{for} \; \kappa \in (0,\kappa_0)$$
(recall that $\alpha[\kappa]$ decreases to 0 when $\kappa$ goes to 0, because of the normalization  \eqref{e_u0_normailized}).

 Inequality \eqref{ze17bis_z} can be \br proved as \eqref{ze17bis} \er just by choosing a (smaller)
$\kappa$ so that  $\delta_0 < \cos n(\theta_++ \alpha[\kappa])$, taking $\rho_0$ large enough so that $\delta_0 \leq \cos n(\theta_++ \alpha[\kappa])-C/\rho_0$ (with \br $C$ as below \er \eqref{zepartialV}) and taking into account that  $\lambda(u)\geq 1$ by \eqref{e_u0_normailized}.

For  inequalities \eqref{ze16_z}, \eqref{ze16bis_z}, \eqref{ze17b_z}, all  the  changes are as for \eqref{ze17_z}. Thus, recall,
\[
\begin{array}{c}
-\partial_\rho V(\rho,\theta,u)+\partial_\rho V(\rho,\theta_1,u)=n\lambda(u)\rho^{n-1}
(\cos n(\theta + \alpha(u))-\cos n(\theta_1 + \alpha(u))) \\ +\sum_{m=1}^{n-1}m\lambda_m(u)\rho^{m-1}(\cos m(\theta + \alpha_m(u))-\cos m(\theta_1 + \alpha_m(u))).
\end{array}
\]
Taking into account the proof of the inequality \eqref{ze17},  consider  two steps. For Step~1, consider any $\epsilon_1\in (0,\theta_-)$ and $\kappa_1>0$  (smaller than the value of $\kappa$ for the previous inequalities) so that $\epsilon_1+\alpha[\kappa_1]<\theta_-$. Redefine $a(\epsilon_1)$ in \eqref{e_ae0} as
\begin{equation}\label{e_ae0_z}
a(\epsilon_1, \kappa_1):=   \hbox{Min}_{u\in [u_0,u_0+\kappa_1]}\{\cos(n(\theta_- -\epsilon_1+\alpha(u)))-\cos(n(\theta_- +\alpha(u)))\}
 \end{equation}
and consider the corresponding bound for $\rho_1$ as in \eqref{e_choice_rho}. For Step 2 and the remainder, the only caution is
to take $\kappa$ small so that $\alpha[\kappa ]<\theta_-/2$.
Indeed, this allows to bound $\cos(n(\theta+\alpha(u)))-
\cos(n(\theta_1+\alpha(u)))$  expanding in a power series \br around $\theta_1+\alpha(u)$ \er as in  \eqref{e_pu2}  but
replacing \br   $\sin (n\theta_-)$ in the lower bound by  $\sin (n\theta_-/2)$ (again, independent of $u$).
\er So, one can proceed as in that proof.

 For \eqref{d}, \eqref{dbis} (and, analogously, \eqref{dd}, \eqref{ddbis}), recall
 \[
\begin{array}{l}
 \rho^{-3/2} (\partial_u V(\rho,\theta_1,u)-\partial_u V(\rho,\theta,u))=
 \\
-\lambda'(u)\rho^{n-3/2}(\cos n(\theta_1+\alpha(u))-\cos n(\theta+\alpha(u)))
\\
n\lambda(u)\alpha'(u)\rho^{n-3/2}(\sin n(\theta_1+\alpha(u))-\sin n(\theta+\alpha(u)))
\\
-\sum_{m=1}^{n-1}\lambda'_{m}(u)\rho^{m-3/2}(\cos m(\theta_1+\alpha_m(u))-\cos m(\theta+\alpha_m(u)))
\\
\sum_{m=1}^{n-1}m\lambda_{m}(u)\alpha'_m(u)\rho^{m-3/2}(\sin m(\theta_1+\alpha_m(u))-\sin m(\theta+\alpha_m(u))).
\end{array}
\]
Whenever $u$ varies in a compact interval, say, $[u_0, u_0+\kappa]$, all the functions of $u$ and their derivatives are bounded; therefore, one obtains a bound type $\leq B \rho^{n-3/2}$ for the absolute value of the right hand side.
Moreover, expanding again in a power series \br around $\theta_1+\alpha(u)$ and $\theta_1+\alpha_m(u)$, \er one obtains a better estimate multiplying by the factor $(\theta_1-\theta)$ close to $\theta_1$ (eventually taking a bigger
$B$), \br thus proving\footnote{This bound \eqref{d} is valid \br for all $\theta \in (-\theta_1,\theta_1)$, with no need of $\epsilon$ in the interval at \eqref{dbis}, \er because we are looking for an upper bound, in contrast with  previous lower bounds.} \eqref{d}. \er Then, the \br bound \eqref{dbis} \er in terms of $\partial_\rho V$ (now in the required interval for $\theta$, which depends on $\epsilon$) \br follows from \eqref{ze17_z}, \er 
as  $\rho_0$ can be chosen bigger again (without altering the other constants) so that $B/\sqrt{\rho_0}<\sqrt{\delta_0} \, \delta$.

Finally, the \br last part \er follows as in Lemma \ref{zl1} taking into account that, once $\epsilon_1, \kappa_1$ satisfying  $\epsilon_1+\alpha[\kappa_1]<\theta_-$ has been chosen, inequalities analog to \eqref{e_auxil1} \br (with $a(\epsilon_1,\kappa_1)$ in \eqref{e_ae0_z} playing the role of $a(\epsilon_1)$) \er hold just replacing $\sin (n\theta_-)$ by \br $\sin n(\theta_--\alpha[\kappa_1])$. \er
 $ \Box$

\begin{remark}{\em
Here, not only the bounds of Lemma \ref{zl1} are improved  by taking into account the $u$-dependence of $V$,
but also the new inequalities \eqref{d}---\eqref{ddbis}
appear. In these inequalities,  the scale factor $\rho^{-3/2}$ (which encourages the corresponding inequality) has been introduced by convenience. We will not use the intermediate inequality  involving  $B$ \br in \eqref{d}, \eqref{dbis} nor in  \eqref{dd}, \eqref{ddbis};  however, they stress \er the simple character of the obtained \br upper \er bound, in comparison with the previous ones.

}\end{remark}

The key in the non-autonomous case is to show that the comparison of energies which leaded to Lemma \ref{zl2} still allows to obtain suitable estimates. With this aim, rewrite the
 conservation law \eqref{e18} as: \begin{equation}\label{e18_z}
\begin{array}{c}
\frac{1}{2}(\dot{\rho}(s)^2+\rho(s)^2\dot{\theta}(s)^2)+V(\rho(s),\theta(s),s) 
= E(s_0)
+ \int_{s_0}^s\partial_u V(\rho(\sigma),\theta(\sigma),\sigma)d\sigma,
\end{array}
\end{equation}
for all $s, s_0\in I=[0,b)$.
In what follows, $s_0$ will be prescribed in the region $[u_0,u_0+\kappa)$ satisfying $\dot \theta({s_0})=0$; so, its total energy is
$E(s_0)=V(\rho(s_0),\theta(s_0),s_0)+ \dot\rho(s_0)^2/2.$
Moreover, for $u_0, \kappa$ as above and $s_1\in [u_0, u_0+\kappa)$
the {\em rectified energy} of its  $\theta(s_1)$-projection $\gamma_{\theta(s_1)}(s)\equiv (\rho(s), \theta(s_1))$ is

\begin{equation}\label{ef1_z}
F (s):=\frac{1}{2}\dot{\rho}(s)^2+V(\rho(s),\theta(s_1),s) -\int_{s_0}^s\partial_u V(\rho(\sigma),\theta(\sigma),\sigma)d\sigma .
\end{equation}

The crucial step is the following  analog to Lemma \ref{zl2}.
\begin{lemma}\label{l2} Choose $0<\theta_-<\theta_+<\pi/(2n)$, and let $A,\rho_0,\kappa>0$ be as in the last part of Lemma \ref{zl1_z}. Let $(\rho(s),\theta(s))$, $s\in [u_0,u_0+\kappa )$ be a solution with $(\rho(s_0),\theta(s_0))\in D[\rho_0,\theta_+]$, $\dot{\theta}(s_0)=0$ and $\dot{\rho}(s_0)\geq 0$ for some $s_0\in [u_0,u_0+\kappa )$. If $s_1\in (s_0,u_0+\kappa)$ satisfies $|\theta(s_0)|<|\theta(s_1)|<\theta_+$ and $\theta(s)$ is 
monotonous on $(s_0,s_1)$,  then $|\theta(s_1)|-|\theta(s_0)|\leq A/\rho(s_0)$.
\end{lemma}
\noindent {\em Proof}.
As in the proof of Lemma \ref{zl2}, by the first equation (\ref{e.14'}) and inequality (\ref{ze17bis_z}),
\begin{equation}\label{zyy_z}
\ddot{\rho}(s)\geq -\partial_\rho V(\rho(s),\theta(s))>\delta_0 n\rho(s)^{n-1}>0, \qquad \qquad \dot \rho(s)>0,
\end{equation}
for some $\delta_0>0$ and all $s\in(s_0,s_1]$.
Using the conservation law (\ref{e18_z}), $F$ in \eqref{ef1_z} is

\begin{equation}\label{g}
F(s)=E(s_0)-\frac{1}{2}\rho(s)^2\dot{\theta}(s)^2
+ V(\rho(s),\theta(s_1),s) -V(\rho(s),\theta(s),s).
\end{equation}
Reasoning  by contradiction as in Lemma \ref{zl2}, if
$|\theta(s_1)|-|\theta(s_0)|>A/\rho(s_0)$
the last part of Lemma \ref{zl1_z} will be applicable and
\begin{equation}\label{vv}
F(s_1)-F(s_0)\leq -V(\rho(s_0),\theta(s_1),s_0)+V(\rho(s_0),\theta(s_0),s_0)<0.
\end{equation}
Thus, \br $\dot F(s_*)<0$ \er for some $s_*\in(s_0,s_1)$ and \br from \eqref{ef1_z} \er
\begin{equation}\label{eq}
\begin{array}{c}
\ddot{\rho}(s_*)<-\partial_{\rho}V(\rho(s_*),\theta(s_1), \br s_* \er ) \qquad\qquad\qquad\qquad\qquad\qquad\qquad\qquad
\\ \qquad\qquad\qquad\qquad -\frac{1}{\dot{\rho}(s_*)}\left(\partial_u V(\rho(s_*),\theta(s_1), \br s_* \er )-\partial_u V(\rho(s_*),\theta(s_*), \br s_* \er ))\right).
\end{array}
\end{equation}
In the case that the last line of this formula is negative, one can drop it and apply  \eqref{ze17_z}, \eqref{ze17b_z}   in order to obtain \br (as in Lemma \ref{zl2}) \er
\begin{equation}\label{e20'}
\begin{array}{rl}
\ddot{\rho}(s_*)<-\partial_{\rho}V(\rho(s_*),\theta(s_*),s_*).
\end{array}
\end{equation}
So, a contradiction with the first equation of (\ref{e.14'}) follows.
Nevertheless, when the last line of \eqref{eq} is positive, a new type of bound 
\bb is required and, as announced in Remark \ref{dificultadnoautonomo} (2), the  hypothesis \eqref{condicionadicional}
will be used. \eb

\bigskip

\noindent {\em Claim. } For any trajectory as above (satisfying \bb $\dot{\rho}( u_0 ) >0$) \eb
and any choice \bb $\delta_0 \in  (0,\cos(n\theta_+)) $ \eb   (as in \bb Lemma \ref{zl1_z}) such that  $\dot{\rho}(u_0)^2>2\delta_0 \rho(u_0)^n$, \eb one has:
$$\dot{\rho}(s_*)^{-1} \br < \er \sqrt{\delta_0^{-1} \, \rho(s_*)^{-3}}.$$

\smallskip

{\em Proof of Claim}. From (\ref{zyy_z}) (recall $n\geq 3)$:
\begin{equation}
\label{e_penultima}
\ddot{\rho}(s)\geq \delta_0 n\rho(s)^{n-1}>0\;\; \Rightarrow\;\;\ddot{\rho}(s)\dot{\rho}(s)\geq \delta_0 n\rho(s)^{n-1}\dot{\rho}(s) \;\; \Rightarrow\;\; \frac{d}{ds}\dot{\rho}(s)^2
\geq 2 \delta_0\frac{d}{ds}\rho(s)^{n}.
\end{equation}
\br Recall that this inequality holds on all $D[\rho_0,\theta_+]$ and, in particular, it implies
\begin{equation}\label{e_ultima}
\dot{\rho}(s)^2>2\delta_0 \rho(s)^n \qquad \forall s\in [u_0,u_0+\kappa ),
\end{equation}
\er \bb as the latter holds  for $s=u_0$ by hypothesis. \eb
\br Integrating \eqref{e_penultima},   \er
\[
\frac{1}{2}(\dot{\rho}(s)^2-\dot{\rho}(s_0)^2)\geq \delta_0 (\rho(s)^n-\rho(s_0)^n),\qquad\forall s\in \br [s_0,s_1] \er
\]
(for $s_0, s_1$ in the lemma). In particular, at the critical point \br $s_*\in(s_0,s_1)$ \er of $F$,
\[
\br \dot{\rho}(s_*)^2\geq \dot{\rho}(s_0 )^2+2\delta_0(\rho(s_*)^n-\rho(s_0)^n)=
\delta_0 \rho(s_*)^3
\left(2\rho(s_*)^{n-3}+\frac{\dot{\rho}(s_0)^2-2\delta_0\rho(s_0)^n}{\delta_0 \rho(s_*)^3}\right). \er
\]
Thus, it is enough to check that the sum  of the terms inside the last big parentheses is bigger than one, \br which holds from   \eqref{e_ultima} plus our overall assumptions (i.e., $1\leq \rho_0 < \rho(s_*)$, $n\geq 3$). \er
 $\Box$

\bigskip
So, when the last line of \eqref{eq} is positive,  the Claim yields:
\begin{equation}\label{ddd}
\begin{array}{c}
\ddot{\rho}(s_*)<-\partial_{\rho}V(\rho(s_*),\theta(s_1),s_*) \qquad\qquad\qquad\qquad\qquad\qquad\qquad\qquad
\\ \qquad\qquad\qquad
-\frac{1}{\sqrt{\delta_0}} \rho(s_*)^{-3/2}
(\partial_u V(\rho(s_*),\theta(s_1),s_*)-\partial_u V(\rho(s_*),\theta(s_*),s_*)).
\end{array}
\end{equation}
As in the previous case (and in the proof of Lemma~\ref{zl2}), the
last paragraph of Lemma~\ref{zl1_z} becomes applicable with $\theta= \theta(s_*)$, $\theta_1= |\theta(s_1)|$ and, using \eqref{d}---\eqref{ddbis}:
\[
\begin{array}{c}
-\frac{1}{\sqrt{\delta_0}} \rho(s_*)^{-3/2}
(\partial_u V(\rho(s_*),\theta(s_1),s_*)-\partial_u V(\rho(s_*),\theta(s_*),s_*))
\\
<  \partial_{\rho} V(\rho(s_*),\theta(s_1),s_*)-\partial_{\rho} V(\rho(s_*),\theta(s_*),s_*).
\end{array}
\]
This last inequality
plus (\ref{ddd}), implies \eqref{e20'} and, thus, the required contradiction.
 $ \Box$

The previous lemmas are enough for a proof as in the autonomous case. Indeed:

\begin{proposition}\label{zl5_z} Under the hypotheses of Prop. \ref{l5}, and for any $u_0\in \R$,  as in Prop. \ref{p_polinomio_no_autonomo}, there exists $\rho_0>1$ and $\kappa_0>0$ such that any solution $\gamma(s)=(\rho(s), \theta(s))$, $s\in [u_0,b), b\in (u_0, u_0+\kappa_0]$ of \eqref{e.14'} satisfying  \br the initial conditions in Prop. \ref{l5}  (formula \eqref{e_condicionesiniciales})   at $u_0$, remains in $D[\rho_0, \theta_+]$. \er 
\end{proposition}
\noindent {\em Proof.} Follow  the proof of Prop. \ref{l5}, just taking into account  the caution to choose $\kappa$ small as in Lemmas \ref{zl1_z} (and, thus, Lemma \ref{l2}) and satisfying, additionally, $\alpha[\kappa]+ \theta_0 <\theta_+$. Indeed, this allows to obtain the
inequalities in Lemma \ref{l_previo}, and only these inequalities plus the general ones in Lemmas \ref{zl1_z}, \ref{l2} are necessary to establish the conclusions of Lemma \ref{l5bis}  and carry out a proof analogous to Proposition \ref{l5}. $ \Box$

\bigskip

\noindent {\em Proof of Theorem \ref{t}.} Reasoning for $k=0$ and
 $0<\theta_0<\theta_+<\pi/(2n)$ as in the statement of the theorem,  take $\kappa_0>0$
and $\rho_0>0$ provided by Prop. \ref{zl5_z} and
consider the inextensible trajectory $\gamma:[u_0,b)\rightarrow \R^2$ under the \br posed initial conditions. \er  
Whenever  its parameter satisfies $s \leq u_0+\kappa_0$, Prop. \ref{zl5_z} ensures that $\gamma$ remains in $D[\rho_0,\theta_+]$. Thus,
the first equation of (\ref{e.14'}) and 
inequality \eqref{ze17bis_z}, imply
\begin{equation}\label{ff'}
\ddot{\rho}(s)=\rho(s)\dot{\theta}(s)^{2}-\frac{\partial
V}{\partial \rho}(\rho(s),\theta(s),u(s))>\delta_0 n\rho(s)^{n-1}.
\end{equation}
Then, the part (i) of  Lemma \ref{e_ODEincomplete} assures that, for the already prescribed value of $\kappa_0$, one can take, if necessary, a bigger $\rho_0$ so that $b<u_0+\kappa_0$ and, thus, all the required trajectories become incomplete.
 $\Box$

\section{Conclusion}

 In this paper, the physical and mathematical fundamentals of the EK conjecture has been clarified, stressing an important link to dynamical systems and potential theory. The solved polynomially bounded case becomes significant in several fields: (i) Relativity, in the foundational basis of gravitational waves, (ii) Classical Mechanics,  as completeness is also relevant in this field and the forces under consideration (coming from a divergence free gradient potential) are the most standard ones in Mechanics, (iii) Dynamical Systems, because the proof develops some ideas which may be of interest in the field \br (namely, how   to control  the oscillations of trajectories, in order to assure their  confinement    in suitable regions), \er and (iv) the theory of complex holomorphic vector fields, as our results  (obtained by avoiding analyticity so that the non-autonomous case is included) makes sense in that field and, prospectively, may introduce new tools and ideas there.
Extensions of these ideas could make sense in recent proposals which develop a more general modeling of pp-waves
\cite{FP,SSS}.
What is more, the EK conjecture introduces the pattern:
$$
\hbox{Source-free dynamics} \Longrightarrow
\left\{
\begin{array}{l}
\hbox{Natural (mathematical) vacuum solutions, or}\\
\hbox{Incompleteness (eventually missed source),}
\end{array}
\right.
$$
which may serve as a paradigm for other parts of Physics, as well as for its mathematical modeling.
\bigskip

\section*{Acknowledgements}
The authors acknowledge warmly  A. Bustinduy (Univ. A. Nebrija, Madrid) and D. Peralta-Salas (ICMAT, Madrid) for discussions on the topic along several years. They are also very grateful to S. Nemirovski (Univ.  Bochum and Steklov Inst. Math.) for his comments, including to point out  the reference \cite{Re}. The careful reading and suggestions by the anonymous referee are also
acknowledged.

{\em Funding: } This work has been partially supported by the grants number MTM2016-78807-C2-2-P
(JLF) and MTM2016-78807-C2-1-P (MS), both of them funded by the Spanish
Ministry of Economy and Competitiveness (MINECO) and the European Regional
Development Fund (ERDF).
%



\end{document}